\documentclass[a4paper, 10pt]{article}
\usepackage{amsmath}
\usepackage{amssymb}
\usepackage{amsthm}
\usepackage{verbatim}
\usepackage{stmaryrd}
\usepackage{mathtools}
\usepackage{graphics}
\long\def\beginpgfgraphicnamed#1#2\endpgfgraphicnamed{\includegraphics{#1}}
\usepackage{anysize}
\usepackage{url}
\usepackage[colorlinks]{hyperref}
\marginsize{3cm}{3cm}{2cm}{2cm}
\newcommand{\prob}{\mathop{\mathrm{Pr}}}
\newcommand{\pex}{\mathop{\mathbb{E}}}
\newcommand{\argmin}{\mathop{\mathrm{argmin}}}
\renewcommand{\leq}{\leqslant}
\renewcommand{\geq}{\geqslant}
\newcommand{\braket}[2]{\langle{#1}\vert{#2}\rangle}
\newcommand{\bra}[1]{\left\langle{#1}\right\vert}
\newcommand{\ket}[1]{\left\vert{#1}\right\rangle}
\newcommand{\weight}[1]{\lvert #1 \rvert}
\newcommand{\norm}[1]{\left| #1 \right|}
\newcommand{\snorm}[1]{\left\| \,#1\, \right\|}

\newcommand{\init}{\preccurlyeq}
\newcommand{\indef}[1]{\textbf{#1}}
\renewcommand{\phi}{\varphi}
\newtheorem{definition}{Definition}[section]
\newtheorem{lemma}[definition]{Lemma}
\newtheorem{theorem}[definition]{Theorem}
\title{The Role of Correlated Noise in Quantum Computing}
\author{Daan Staudt \\ ILLC, University of Amsterdam \\ dstaudt@science.uva.nl}

\begin{document}
\maketitle
\begin{abstract}
This paper aims to give an overview of the current state of fault-tolerant
quantum computing, by surveying a number of results in the field. We show
that thresholds can be obtained for a simple noise model as first proved in
\cite{AharonovFTQC, KitaevFTQC, KnillFTQC}, by presenting a proof for
statistically independent noise, following the presentation of Aliferis,
Gottesman and Preskill \cite{aliferis_2005}. We also present a result by
Terhal and Burkard \cite{terhal_2005} and later improved upon by Aliferis,
Gottesman and Preskill \cite{aliferis_2005} that shows a threshold can
still be obtained for local non-Markovian noise, where we allow the noise
to be weakly correlated in space and time. We then turn to negative results,
presenting work by Ben-Aroya and Ta-Shma \cite{ben-aroya_2009} who showed
conditional errors cannot be perfectly corrected. We end our survey by
briefly mentioning some more speculative objections, as put forth by Kalai
\cite{kalai_detrimental_2008, kalai_propagation, kalai_fail}.
\end{abstract}

%
%
\section{Introduction}
\label{intro}
We have come to take for granted that our modern (classical) computers can
perform complex computations for hours, days or weeks on end without failing.
We say that such implementations of classical computation are essentially
perfect. For the successful implementation of a \emph{quantum} computer,
however, we will have to guard against noise impacting our computation. This
paper discusses several results in the area of quantum error-correction and
fault-tolerant quantum computing (FTQC). We assume familiarity with the
basic principles of quantum computing (see, e.g., \cite{nielsen} or
\cite{qcnotes} for an introduction) as well as some knowledge of linear
algebra and discrete mathematics.

When dealing with classical computation, we often use the Turing Machine (TM)
model of computation. There is a quantum analogue to the TM called the Quantum
Turing Machine (QTM), but it is highly complex and not nearly as intuitive for
quantum computing as the TM is for classical computing. Instead of the QTM model,
it is standard when dealing with quantum computing to look at the \indef{quantum
circuit} model. As with classical (Boolean) circuits, a quantum circuit is built
up from a variety of gates. Instead of logical (Boolean) gates such as AND or OR,
quantum circuits contain quantum gates; unitary operations on a fixed number
of qubits, usually 1, 2 or 3.

These gates can be executed sequentially, corresponding to the ordinary product
of the unitaries. They can also be executed in parallel, on different input
qubits, corresponding to the tensor product of the unitaries. Naturally there
are infinitely many different quantum gates. We usually assume that we make use
of only a finite number of gates which suffice to approximate any unitary
gate. It is not important to our discussions exactly which so called
\indef{universal set of gates} is used, as long as the set is finite. Common
examples of universal sets of gates include the set containing the Hadamard
gate ($H$) and the Toffoli gate ($T$) or the set containing the Hadamard gate,
the CNOT gate and the $\pi/8$ gate, $\big(\begin{smallmatrix} 1 & 0 \\ 0 & e^{i
\pi / 4} \end{smallmatrix}\big)$. The longest time required to execute any of
the gates in the universal set is called the \indef{fundamental gate time}.

The input to a circuit is conceptually divided into a number of
\indef{registers}. We sometimes need to make use of temporary qubits, prepared
in a fixed known state (usually $\ket{0}$), that we dispose of immediately
after use. Such temporary qubits that are not considered part of any input
register are called \indef{ancilla qubits}. The quantum circuit we wish to
implement may contain a number of measurements of one or more qubits. However,
we may and will assume that those measurements are postponed until the very
end of the circuit, see, e.g., \cite[Section 4.4]{nielsen}.

The noise we have to protect our quantum computations from can be viewed in
two ways. First we consider noise as impacting qubits while they are not being
acted upon, the so-called \indef{storage errors}, or \indef{errors} for short. We
then look at noise as something that impacts the performance of our quantum
gates. Such noise is said to introduce \indef{faults} in our gates. For a
description of the various types of noise that can impact a quantum computation
and their physical motivations, see, e.g., \cite[Sections 1.2.2.1,
1.2.3.1]{aliferis_2011}. The term \indef{fault-tolerant quantum computing}
(FTQC) is used to describe quantum computers that are capable of
dealing with faults and errors without yielding incorrect answers. This is
quite different from the way we consider \emph{classical} computation, which we
called essentially perfect. A computer being perfect means that it is
insusceptible to noise, rather than being capable of mitigating its harmful
influences.

Let us first consider storage errors in the classical case. When we limit
ourselves to a single bit, the damage any noise can do is quite
restricted. Either the bit is left intact or it is flipped, i.e., 0 becomes 1 or
1 becomes 0. When we consider a string of bits, however, more variation is
possible. Any combination of bits in the string could be flipped. There could
be some imaginary adversary that decides which bits to flip, but for now we
will limit ourselves to a simpler noise model. In this model each bit of a
binary string $x$ is flipped with some probability $p$ independent of the
others. We refer to this model as \indef{independent noise} or a bit-flip
channel. Suppose this noise were to impact our bitstring $x$ of length $n$.
Then with probability $(1-p)^n$ the string is left intact and with probability
$1 - (1-p)^n$ at least one bit is flipped. Without error correction we have no
way to detect which bit or bits have been flipped and so with probability $1 -
(1-p)^n$ we cannot recover the intended state of the string.

The solution is to encode the intended state, adding redundant information. That
way we can tolerate a portion of the information being lost while still being
able to recover the original string. Fortunately in the classical world we can
clone information, so we can, for example, copy each bit in our string several
times in the hope that few of the bits will be damaged. We can then decide
what the original bit was by taking the majority value of our copies. More
formally we define an \indef{encoder} $C: \{0,1\} \to \{0,1\}^3$ that encodes
one bit into three bits by $C(0) = 000$ and $C(1) = 111$. We refer to bits we
wish to encode (0 and 1) as \indef{logical bits} and to the bits that are
stored or transmitted ($C(0)$ and $C(1)$) as \indef{physical bits}. To recover
the logical bits from the physical bits we need a \indef{decoder} $D: \{0,1\}^3
\to \{0,1\}$ which will output the majority value of its input. Define the
\indef{Hamming weight} of $x \in \{0, 1\}^n$ as $\norm{x} = \norm{\{x_i \mid 1
\leq i \leq n \text{ and } x_i = 1 \}}$, i.e., the number of 1s in $x$. Now
for $x \in \{0,1\}^3$ we define $D(x) = 0$ if $\norm{x} \leq 1$ and $D(x) =
1$ if $\norm{x} \geq 2$.

Now let us analyze what happens when $C(0)$ or $C(1)$ are exposed to the
independent noise. So long as at most one bit of $C(0)$ or $C(1)$ is flipped,
our decoder $D$ will return us to the intended state. The intended state
becomes unrecoverable when at least 2 bits are flipped. Because each bit is
flipped with probability $p$ independent of the others, the probability that
at least 2 bits are flipped is $3(1-p)p^2 + p^3$, i.e., the probability that 2
or 3 bits are flipped. This by itself does not mean that we have increased the
probability of recovering our string. For that to have happened it must be that
$3p^2 - 2p^2 < p$. This will certainly be the case if $3p^2 < p$, which happens
when $p < 1/3$. We refer to $1/3$ as a \indef{threshold value}. Once we can get
the noise rate below it we are certain we can encode information to increase the
chance of recovering from the noise. In fact we can bring the probability of
recovering the intended state arbitrarily close to 1 by repeating the encoding
procedure an arbitrary number of times: we can consider each of the bits of
$C(0)$ and $C(1)$ as logical bits themselves and encode them using $C$. This
process is called \indef{concatenation}.

In the quantum case even a single qubit can be exposed to a \emph{continuum}
of different errors. This is easy to see when we write a qubit $\ket{\phi}$ as
$\alpha \ket{0} + \beta \ket{1}$. So long as $\norm{\alpha}^2 + \norm{\beta}^2
= 1$ we still have a qubit and there are uncountably many pairs $(\alpha,
\beta)$ that satisfy the equation. For now we will assume that noise acting on
a single qubit is some \emph{unitary operator}. It is common to write such
unitaries in the Pauli basis, given by
\begin{equation*}
I = \begin{pmatrix}
1 & 0\\
0 & 1\\
\end{pmatrix}, \quad X = \begin{pmatrix}
0 & 1\\
1 & 0\\
\end{pmatrix}, \quad Y = \begin{pmatrix}
0 & -i\\
i & 0\\
\end{pmatrix}\text{ and } Z = \begin{pmatrix}
1 & 0\\
0 & -1\\
\end{pmatrix}.
\end{equation*}
Every 2-by-2 matrix can be written as a linear combination of these
four matrices. In particular we refer to unitary errors of this form as
\indef{Pauli errors}. We note that $Y = i X Z$, so apart from a global phase of
$i$ we can act as though $Y = X Z$.

Just as before we must encode our quantum states if we wish to guard against
noise. We formalize the concept of encoding quantum states as follows.
\begin{definition}
Given a state space $\mathcal{N}$ with $\dim(\mathcal{N}) = 2^n$, we call a
space $\mathcal{M} \subseteq \mathcal{N}$ an $\llbracket n,k \rrbracket$
Quantum Error Correcting Code (QECC) encoding $k$ qubits into $n$ qubits if
$\dim(\mathcal{M}) = 2^k$. Associated with each QECC is a map from $k$-qubit
states $\ket{x}$ to their encoded $n$-qubit states $\ket{\overline{x}}$.
\end{definition}
We call a vector in $\mathcal{M}$, i.e., an encoded quantum state, a
\indef{codeword} and we refer to $\mathcal{M}$ as a \indef{QECC} or as the
\indef{code space}. As an example, we will analyze what happens when a qubit
that is encoded using a $\llbracket 9,1 \rrbracket$ QECC called the Shor code
is exposed to a unitary storage error. The Shor code encodes $\ket{0}$ as
$\frac{1}{\sqrt{8}}(\ket{000} + \ket{111})^{\otimes 3}$ and $\ket{1}$ as
$\frac{1}{\sqrt{8}}(\ket{000} - \ket{111})^{\otimes 3}$. For $b \in \{0,1\}$
we will denote the encoding of $\ket{b}$ by $\ket{\overline{b}}$.

Let us consider what happens when an $X$ error hits one of the qubits of
$\ket{\overline{b}}$. It can hit any of the nine qubits, so for $1 \leq k
\leq 9$, let $X^k$ be the operator that applies $X$ to the $k$'th qubit of
$\ket{\overline{b}}$ and $I$ to the others. So if an $X$ error hits $\ket{
\overline{b}}$, we are left with $X^k\ket{\overline{b}}$ for some $k$. We can
detect which of the nine qubits was subjected to the $X$, i.e., we can
determine $k$. This can be done without collapsing the state, see, e.g.,
\cite[Chapter 10]{nielsen} for details. We can store the location $k$ of
the affected qubit using four ancilla qubits. By convention if no $X$ error
has occurred we let $k = 0$ and we define $X^0$ as the operator that applies $I$
to all qubits in the state. If $\ket{\overline{b}}$ is struck by a $Z$ error,
then one of the three blocks $(\ket{000} \pm \ket{111})$ will have a different
sign than the other two. Note that it does not matter which of the qubits in
the block was hit by the error, as the effect is the same. So we can for $1
\leq \ell \leq 3$ define $Z^\ell$ to be the operator that applies $Z$ to some
qubit in the $\ell$'th bock of $\ket{\overline{b}}$ and $I$ to the others. So
if a $Z$ error hits we are left with $Z^\ell \ket{\overline{b}}$ for some
$\ell$. As with $X$ errors we can detect which block was hit by a $Z$ error,
i.e., we can determine $\ell$. This information can be stored using two more
ancilla qubits. As with $X$, we let $\ell = 0$ if no $Z$ error has occurred
and define $Z^0$ to be the operator that applies $I$ to all the qubits in the
state.

So if our state $\ket{\overline{b}}$ is hit by an $X$ or $Z$ error, i.e., has
turned into $X^k Z^\ell \ket{\overline{b}}$ for some $k \in \{0, \ldots 9\}$ and
$\ell \in \{0, \ldots, 3\}$ we can detect these errors and write $k$ and $\ell$
into ancilla qubits to obtain $X^k Z^\ell \ket{\overline{b}}\ket{k}\ket{\ell}$.
This procedure is called \indef{error detection} and we refer to the
pair $(k,\ell)$ as the \indef{error syndrome}. We now measure the ancilla qubits
to obtain the error syndrome. We correct the errors by applying another $X^k$
to the state and applying a $Z$ to some qubit in the $\ell$'th block, say the
first. This is called \indef{error correction}. As discussed before, a $Y$
error hitting a qubit is the same, modulo a global phase of $i$, as both an $X$
and a $Z$ error hitting that qubit. So if we define $Y^k$ as the operator that
applies $Y$ to the $k$'th qubit and $I$ to the others we can say that $Y^k \ket{
\overline{b}} = i X^k Z^\ell \ket{\overline{b}}$, where $\ell$ is the block
containing the $k$'th qubit. We can perform the error detection and error
correction steps to obtain, after discarding the ancilla qubits, the state
$i\ket{\overline{b}}$. Note that measuring this state gives the same probability
distribution as measuring $\ket{\overline{b}}$, so we are safe to ignore this
global phase and say that we can also correct $Y$ errors. Also note that we can
trivially correct an $I$ error, as this leaves the state intact. When we
perform the error-detection step, the error syndrome will be $(0,0)$ and
performing $X^0$ and $Z^0$ has no effect, so the state remains correct. After
the error-detection and error-correction steps the Shor code discards the ancilla
qubits that held the error syndrome. This is a way to remove from the system the
entropy introduced by the errors. As such the Shor code requires a constant fresh
supply of properly prepared ancilla qubits. In fact such a constant fresh supply
of ancillas is a prerequisite for implementing \emph{any} QECC.

Furthermore, since any 2-by-2 matrix $M$ can be written as a linear combination
of the Pauli matrices, i.e., $M = \alpha_I I + \alpha_X X + \alpha_Y Y +
\alpha_Z Z$, we can correct any error hitting a single qubit. To see this note
that when the $k$'th qubit of a state $\ket{\overline{b}}$ is subjected to $M$
we obtain after the error-detection step the following state, where $\ell$ is the
block containing the $k$'th qubit.
\begin{equation*}
\alpha_I \ket{\overline{b}}\ket{0}\ket{0} + \alpha_X X^k\ket{\overline{b}}
\ket{k}\ket{0} + \alpha_Y X^k Z^\ell \ket{\overline{b}}\ket{k}\ket{\ell} +
\alpha_Z Z^\ell \ket{\overline{b}}\ket{0}\ket{\ell}
\end{equation*}
Note that we ignored the global phase of $i$ for the $Y$ error. Measuring the
ancilla qubits, this state collapses to one of the four terms. Each of those we
can correct to recover $\ket{\overline{b}}$.

Faults can arise as the result of an imperfect implementation of a
quantum gate. They can also be the result of an imperfect preparation of a
register qubit. The input to a quantum circuit is given classically and the
quantum computer must encode this classical state into the register qubits
before the circuit can be executed. Errors may also arise as the result of
imperfect measurements. We call a gate that performs a different operation than
intended a \indef{faulty gate}. Noise is to blame for gates being faulty, so we
say that noise introduces \indef{faults} in gates. As hinted to above, we can
use QECCs to protect gates from faults induced by noise. In particular we can
prevent gates from spreading errors in their input to errors in their output
too much. Such implementations of gates are called \indef{fault-tolerant}
(implementations of) gates. 

The way circuits are constructed, with both parallel and sequential executions
of gates, it is very common for some qubits or even entire registers to remain
\indef{resting} for large portions of the circuit's execution. Resting in this
case means that there are no gates acting on them. The longer a qubit is
resting, the more likely it is to be hit by storage errors. We can however
consider such qubits to be acted upon by identity gates ($I$) and we can create
fault-tolerant implementations of identity gates to guard against storage
errors. This shows how we can think of errors as faults. We can also think of
faults as (storage) errors by thinking of a faulty gate as an ideal gate
followed by some (not necessarily unitary) error operator. That error operator
can then be seen as causing a storage error. This illustrates how we can think
of noise as something that causes errors or as something that causes faults
(or both), whichever way of thinking is more convenient for us at any given
time. 

The simplest noise models in terms of analysis are independent noise models,
where each qubit is hit by an error independent of the others. It is generally
believed, however, that such noise models are not physically realistic. It is
assumed that in physically realistic models the noise will be correlated, either
in time, in space or in both. It may also be possible that the noise does not
act the same on each term of a state in superposition, something that we do
assume with the independent noise model. This paper aims to give an overview
of the current state of FTQC, of noise models for which we have threshold
results as well as types of noise for which no threshold can be obtained.

This paper is structured as follows. We start with some positive results, i.e.,
showing that FTQC is possible provided the noise levels are low enough. In
Section \ref{frame} we show this for an independent noise model using the
framework and method presented in \cite{aliferis_2005}. In Section
\ref{nonmarkov} we show it for a ``local non-Markovian'' noise model, as first
done by Terhal and Burkard in \cite{terhal_2005} and later improved upon by
Aliferis, Gottesman and Preskill in \cite{aliferis_2005}. Then in Section
\ref{object} we turn to some negative results by Ben-Aroya and Ta-Shma,
who showed in \cite{ben-aroya_2009} that certain types of errors cannot be
corrected by any QECC, although some errors can be \emph{approximately}
corrected. In Section \ref{longtail} we turn more speculative objections to
FTQC as put forth by Kalai in \cite{kalai_detrimental_2008, kalai_propagation, 
kalai_fail}. Finally we conclude in Section \ref{conclusion}.

%
%
\section{A threshold result for independent noise}
\label{frame}
In this section we will describe a general framework for fault-tolerant quantum
computation in the face of independent noise. This result was first proved in
\cite{AharonovFTQC, KitaevFTQC, KnillFTQC}, but here we will follow the
presentation of \cite{aliferis_2005}. The next section will deal with a more
challenging noise model.

The goal is to create (and prove correct) a fault-tolerant implementation of an
arbitrary quantum circuit. That circuit we shall refer to as the \indef{ideal
circuit} and denote by $M_0$. We start by dividing $M_0$ into a set of
\indef{locations}, each corresponding to a single gate, qubit preparation or
measurement in the circuit. Note that we consider a resting qubit (i.e., one
that is not being acted upon by a gate) to be acted upon by the identity gate.
Thus each time interval where a qubit is resting is divided into a number of
locations corresponding to identity gates. Note that we may treat each location
that corresponds to the application of a gate as corresponding to a time
interval of length $t_0$, the fundamental gate time.

Now the QECC comes into play. Let $C$ be a QECC that encodes one (logical)
qubit into $m$ (physical) qubits. We refer to a set of $m$ qubits that are
the encoding of a single qubit by $C$ as a \indef{1-block}. We will encode (here
the term is used informally) each location into a group of locations called a
\indef{rectangle}. There will be rectangles for preparing register qubits, qubit
measurements and the application of gates. When all the locations in $M_0$ are
replaced by rectangles we obtain a new circuit which we shall call $M_1$.

In our ideal circuit $M_0$ we will prepare register qubits in some basis, say
the computational basis. Where in $M_0$ we would simply be supplied with a
$\ket{0}$ qubit, in $M_1$ we will need an encoded $\ket{0}$, $C(\ket{0})$. A
\indef{qubit preparation rectangle} is thus a rectangle that provides us with
$C(\ket{0})$. Furthermore we assume that the rectangle contains circuitry for
performing the error-detection and error-correction steps as described in
Section \ref{intro} after the $C(\ket{0})$. We refer to such error-detecting
and error-correcting circuitry that corrects errors in a 1-block as a
\indef{1-EC}, for Error Correction. Note that we will not replace the
preparation of the ancilla qubits used to hold the error syndrome by the QECC,
only the preparation of register qubits.

Similarly we must be able to measure the logical value of a 1-block in $M_1$,
in other words we must be able to decode a 1-block to obtain the measurement
outcome had we measured in $M_0$. For this operation we use \indef{qubit
measurement rectangles}. Depending on the QECC used and the rectangle design
this might be as simple as measuring each qubit in the 1-block and then taking
the (recursive) majority. Since the measurement outcomes are classical and we
can classically derive the logical output from the measurements there is no
need for a 1-EC in the measurement rectangles. Naturally we assume that
classical information storage and computation is perfect.

Each gate in a circuit is replaced by a \indef{gate application rectangle},
which consists of a fault-tolerant implementation of the gate, called a
\indef{1-Ga} for Gate, followed by a 1-EC. Depending on the code used we may
have to require that the gates of the ideal circuit $M_0$ are gates in a
particular universal set of gates. In this discussion however we will not fix
such a set. We call a group of locations in $M_1$ that are the encoding of a
single location in $M_0$, i.e., that make up the rectangle for that location
in $M_0$, a \indef{1-Rectangle} or \indef{1-Rec} for short.

The procedure of encoding a single (logical) qubit into a rectangle consisting
of (physical) qubits can be repeated many times, by considering the qubits that
make up the rectangles as logical qubits themselves and replacing each by a
rectangle as before. Thus we can have that each location in $M_0$ is encoded by
a rectangle consisting of locations in $M_1$ that are each encoded by a
rectangle consisting of locations in $M_2$ and so on. To reason about such
recursive encodings we shall extend our definitions somewhat.

A set of qubits in $M_r$ that are the (concatenated) encoding of a single qubit
in $M_{r-s}$ is called an \indef{$s$-block} in $M_r$. We have already seen a
1-block in $M_1$, which corresponded to a single qubit in $M_{1-1} = M_0$ and
by the nature of $C$ thus consisted of $m$ qubits. At this level we can say the
1-block contains physical qubit that encode a single logical qubit of $M_0$.
Similarly a 1-block in $M_2$ will be the encoding of a single qubit in $M_{2-1}
= M_1$ and will also be $m$ qubits. From this perspective we can consider the
1-block to consist of the physical qubits that encode a single logical qubit
of $M_1$. A 2-block in $M_2$ however will be the encoding of a single qubit
in $M_{2-2} = M_0$, which corresponds to $m$ qubits in $M_1$, each of which
becomes $m$ qubits in $M_2$, thus the size of a 2-block in $M_2$ is $m^2$
qubits. In general an $r$-block in $M_r$ consists of $m^r$ qubits. Here we
can say that the $m^2$ qubits of $M_2$ are the physical qubits for $m$ logical
qubits in $M_1$, which are themselves physical qubits for a single logical
qubit in $M_0$. So what we call physical or logical qubits in $M_k$ depends on
whether we look `down' to $M_{k+1}$ or `up' to $M_{k-1}$.

Similarly we call a group of locations in $M_r$ that are the (concatenated)
encoding of a single location in $M_{r-s}$ an \indef{$s$-Rec} in $M_r$. We have
seen that a 1-Rec in $M_1$ is a rectangle as it corresponds to a single
location in $M_{1-1} = M_0$. Similarly a 1-Rec in $M_2$ would correspond to a
rectangle for a location in $M_1$ and a 2-Rec in $M_2$ corresponds to the set
of rectangles for the locations in $M_1$ that make up a rectangle for a single
location in $M_0$. In general it may be difficult to calculate the precise
number of locations in an $r$-Rec in $M_r$, but if we let $L$ be the maximum
number of locations in a rectangle we can give an upper bound as $L^r$. We
could also generalize our definitions of a 1-EC and a 1-Ga, but we will not
often need to refer to $s$-ECs or $s$-Gas in this discussion. Figure \ref{fig:mk}
illustrates some of these key definitions.

\begin{figure}[htbp]
\begin{centering}
\beginpgfgraphicnamed{simlevels}
\newcommand{\rec}[3]{
  \begin{scope}[xshift=#1cm,yshift=#2cm,scale=#3]
    \draw[dashed] (0,0) rectangle (3,2);
    \coordinate (p1) at (0,0);
    \coordinate (p3) at (3,2); 

    \draw (1.5,0.5) rectangle (2.5,1.3);
    \draw (1.3,0.7) -- (1.5,0.7);
    \draw (1.3,1.1) -- (1.5,1.1);
    \draw (2.5,0.7) -- (2.7,0.7);
    \draw (2.5,1.1) -- (2.7,1.1);

    \draw (0.4,1.3) rectangle (0.8,1.7);
    \draw (0.2,1.5) -- (0.4,1.5);
    \draw (0.8,1.5) -- (1.0,1.5);
    \coordinate (q1) at (0.4,1.3);
    \coordinate (q3) at (0.8,1.7);

    \draw (0.6,0.3) rectangle (1.0,0.7);
    \draw (0.4,0.5) -- (0.6,0.5);
    \draw (1.0,0.5) -- (1.2,0.5);
  \end{scope}
}
\begin{tikzpicture}[yslant=0.2,xslant=0,scale=0.65]
  \draw[black,thick] (6,6) rectangle (13,11);
  \rec{7}{8}{1}
  \coordinate (p1k) at (p1);
  \coordinate (p3k) at (p3);
  \coordinate (q1k) at (q1);
  \coordinate (q3k) at (q3);

  \draw[black,thick] (14,0) rectangle (21,5);
  \draw[dashed] (15.3,0.7) rectangle (19,3);
  \coordinate (s1k1) at (15.3,0.7);
  \coordinate (s3k1) at (19,3);
  \rec{16}{1}{0.7}
  \coordinate (p1k1) at (p1);
  \coordinate (p3k1) at (p3);
  \coordinate (q3k1) at (q3);

  \draw[dashed] (p1k) -- (s1k1);
  \draw[dashed] (p3k) -- (s3k1);
  \draw[dashed] (q1k) -- (p1k1);
  \draw[dashed] (q3k) -- (p3k1);

  \fill[white,fill opacity=0.7] (14,0) rectangle (21,5);
  \draw[black,thick] (14,0) rectangle (21,5);
  \draw[dashed] (15.3,0.7) rectangle (19,3);
  \rec{16}{1}{0.7}

  \draw[-latex] (15,11) node[right]{$M_k$}
         to[out=180,in=45] (13,11);
  \draw[-latex] (15,10) node[right]{$k$-Rec in $M_k$}
         to[out=180,in=45] (p3k);
  \draw[-latex] (15,9) node[right]{Logical gate in $M_k$}
         to[out=195,in=15] (q3k);
  \draw[-latex] (23,5) node[right]{$M_{k+1}$}
         to[out=180,in=45] (21,5);
  \draw[-latex] (23,4) node[right]{$(k+1)$-Rec in $M_{k+1}$}
         to[out=195,in=45] (19,3);
  \draw[-latex] (23,3) node[right]{1-Rec in $M_{k+1}$}
         to[out=205,in=30] (p3k1);
  \draw[-latex] (23,2) node[right]{Physical gate in $M_{k+1}$}
         to[out=200,in=5] (q3k1);
\end{tikzpicture}
\endpgfgraphicnamed
\end{centering}
\caption{The relation between the level-$k$ encoding of a circuit and the
level-$(k+1)$ encoding of the same circuit. Note that a single `logical' gate
in $M_k$ is encoded by a 1-Rec in $M_{k+1}$. This `logical' gate takes a single
qubit as input, so the corresponding 1-Rec takes a 1-block as input. The gates
that make up each 1-Rec in $M_{k+1}$ are called `physical' gates.}
\label{fig:mk}
\end{figure}
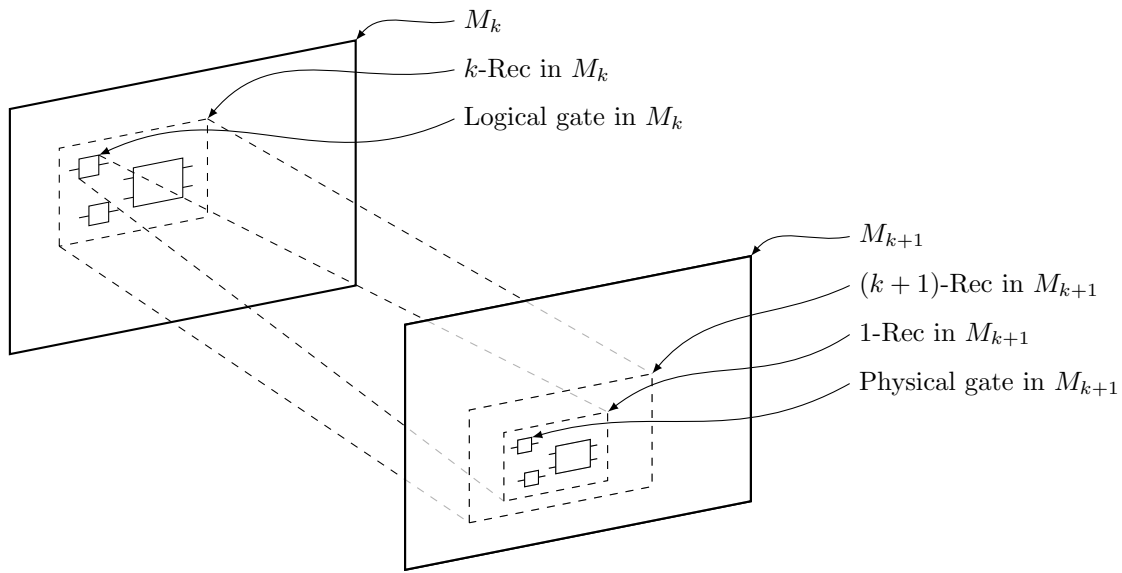

Since each 1-Rec in a circuit ends in a 1-EC, with the exception of measurement
rectangles which we assume to only occur at the very end of the circuit, each
1-Rec is also immediately preceded by a 1-EC. We call a 1-Rec together with
the 1-EC that immediately precedes it a \indef{1-exRec} for `extended
rectangle'. Similarly we can define an \indef{$s$-exRec} in $M_r$ as an $s$-Rec
in $M_r$ with its preceding $s$-EC.

We are assuming that $C$ is a QECC that can correct a single unitary error,
which means we can construct the 1-ECs and 1-Gas in 1-Recs such that the
following conditions are met:
\label{cond15}
\begin{enumerate}
\item If a 1-EC contains at most one fault, then it takes any pure state input
to an output in the code space.
\item If a 1-EC does not contain a fault, then it takes any pure state input
with at most one error to an output with no errors.
\item If a 1-EC contains at most one fault, then it takes a pure state input
with no errors to an output with at most one error.
\item If a 1-Ga contains no fault, then it takes a pure state input with at
most one error to an output with at most one error in each output block.
\item If a 1-Ga contains at most one fault, then it takes a pure state input
with no errors to an output with at most one error in each output block.
\end{enumerate}
See \cite{aliferis_2011} for a general discussion of rectangle design and
\cite[Sections 7 and 8]{aliferis_2005} for an explicit construction of
rectangles that satisfy these conditions.

We say that a 1-exRec is \indef{good} if it is hit by at most one fault and that
it is \indef{bad} if it is hit by at least two. The idea is that a good 1-exRec
will leave at most one error in its output. We call two bad 1-exRecs
\indef{independent} if they do not overlap, i.e., do not share a 1-EC, or if
they do overlap and the first 1-exRec would still contain at least two fault
if we do not count the faults in the shared 1-EC. We define goodness and
badness for higher levels of concatenation recursively. A $k$-exRec is good
if it contains at most one bad $(k-1)$-exRec and bad if it contains at least
two. Analogously to the 1-exRec case we call two bad $k$-exRecs independent
if they do not overlap, i.e., do not share a $k$-EC, or if they do overlap and
the first $k$-exRec would still contain at least two bad $(k-1)$-exRecs if
we do not count the $(k-1)$-exRecs in the shared $k$-EC.

Our strategy for proving the threshold result consists of three stages. First
we show that if a $k$-exRec is good then the $k$-Rec it contains will be
`correct'. Secondly we show that if all $k$-exRecs in $M_k$ are good, the
probability distribution of a measurement of $M_k$ will be the same as that of
$M_0$. Finally we show that the number of bad $k$-exRecs decreases doubly
exponentially as the level of encoding ($k$) increases polynomially. We
conclude with the threshold result.

For the first step in our proof we need to define what it means for a $k$-Rec
to be correct. To this end we introduce the concept of an ideal $k$-decoder in
$M_k$, which we define recursively. An \indef{ideal 1-decoder} in $M_1$ takes
a 1-block as input, performs the error-detection and error-correction steps,
i.e., a 1-EC, and outputs a single decoded qubit.  An ideal \indef{$k$-decoder}
in $M_k$ takes a $k$-block as input and first runs ideal $(k-1)$-decoders on
each of the $(k-1)$-blocks of its input and then uses an ideal 1-decoder on the
resulting 1-block. The decoder is called ideal because we assume that it
contains no faults, hence it is only a theoretical device. Note that these
decoders are not part of the actual fault-tolerant circuit, they are only used
for the analysis.

We can now say that a $k$-Rec for the application of a gate is \indef{correct}
if the $k$-Rec followed by an ideal $k$-decoder is equivalent to the ideal
$k$-decoder followed by the ideal gate it is meant to implement. A $k$-Rec
for qubit preparation is called correct if the $k$-Rec followed by the ideal
$k$-decoder is equivalent to the qubit preparation the $k$-Rec is meant to
implement. Finally a $k$-Rec for qubit measurement is correct if the $k$-Rec
is equivalent to the ideal $k$-decoder followed by the measurement the $k$-Rec
is meant to implement. Thus we can see that a correct $k$-Rec allows its
output state to be successfully decoded by some ideal decoder. We are now
ready to prove our first lemma.

\begin{lemma}[{\cite[Lemma 3]{aliferis_2005}}]
\label{ali3}
Assume conditions 1-5. For $k \geq 1$, if a $k$-exRec is good then the $k$-Rec
it contains is correct.
\end{lemma}
\begin{proof}
We prove this by induction on $k$. For $k = 1$ we first consider a 1-exRec
for a gate application. Because the 1-exRec is good it contains at most one
fault. If it contains no faults the result is immediate. If it contains one
fault we make a case distinction on the location of the fault.
\begin{itemize}
\item If the fault is in one of the 1-ECs in front of the 1-Rec, then by
condition 3 its output contains at most one error. The output of the other
1-ECs is in the code space by condition 1. So the pure state inputs to the
1-Ga contain no errors, i.e., they are all in the code space. Now by condition
4 the output of the 1-Ga contains at most one error in each output block and
by condition 2 this error is corrected by the 1-ECs that follow the 1-Ga.
\item If the fault is in the 1-Ga, then the 1-ECs preceding it have all output
codewords by condition 1. By condition 5 therefore the output of the 1-Ga has at
most one error in each output block, which is corrected by the 1-ECs following
it by condition 2.
\item If the fault is in one of the 1-ECs after the 1-Ga, then the 1-ECs
preceding the 1-Ga have all output codewords by condition 1. Now by condition 4
the output of the 1-Ga has no errors. By condition 3 the output of the 1-Rec
now contains at most one error.
\end{itemize}
A similar argument goes for a 1-Rec for qubit preparation. Either the fault
lies in the preparation, in which case it is corrected by the 1-ECs that
follow it, or the fault is in one of the 1-ECs in which case condition 3
ensures the output contains at most one error. For 1-Recs for qubit measurement
the fault can only lie in one of the preceding 1-ECs, in which case the ideal
1-decoder will correct it.

We only show the inductive step for $(k+1)$-exRecs for gate applications, those
for qubit preparations and measurements are done in a similar fashion. We
need to show that $(k+1)$-exRecs followed by an ideal $(k+1)$-decoder are
equivalent to the $(k+1)$-ECs followed by an ideal $(k+1)$-decoder followed by
the gate the $(k+1)$-Rec is meant to implement. By the definition of an ideal
$(k+1)$-decoder we can view it as a number of $k$-decoders followed by a
1-decoder. Note that each such $k$-decoder is preceded by a $k$-Rec, namely
those the $(k+1)$-Rec is made of. Using the induction hypothesis we can move
the $k$-decoders in front of these $k$-Recs, leaving them as the ideal 1-Recs
they are meant to implement. Now by the base case each 1-Rec followed
by a 1-decoder is equivalent to a 1-decoder followed by the gate the 1-Rec is
meant to implement.

Now our circuit has the following shape. First there are a number of
$(k+1)$-ECs, then the $k$-decoders, then 1-decoders and finally the gate our
$(k+1)$-exRec was meant to implement. But again by the definition of ideal
decoders, this is equivalent to the $(k+1)$-ECs followed by a $(k+1)$-decoder
followed by the gate the $(k+1)$-exRec was meant to implement. This completes
the proof.
\end{proof}

For the second part of the proof we use the following short lemma.

\begin{lemma}[{\cite[Lemma 4]{aliferis_2005}}]
\label{ali4}
Assume conditions 1-5. If all $k$-exRecs in $M_k$ are good, then $M_k$ has the
same probability distribution on its outcome as $M_0$.
\end{lemma}
\begin{proof}
By Lemma \ref{ali3} all $k$-exRecs being good implies that all the $k$-Recs
contained in them are correct. In particular all qubit-preparation $k$-Recs
output at most one error and all the $k$-Recs for gate applications do no spread
this error. Thus at most one error per block arrives at the $k$-Recs for qubit
measurement at the end of the circuit and these $k$-Recs perform the measurement
faithfully, i.e., without faults.
\end{proof}

Our last lemma will show that there are few bad $k$-exRecs.

\begin{lemma}[{\cite[Lemma 2]{aliferis_2005}}]
\label{ali2}
Let $A$ be the largest number of pairs of locations in any 1-exRec. Assuming a
noise model where faults occur in a location within a $k$-exRec with probability
$\epsilon$ independently, the probability $\epsilon^{(k)}$ that a $k$-exRec is
bad satisfies
\begin{equation*}
\epsilon^{(k)} \leq \frac{(A\epsilon)^{2^k}}{A}.
\end{equation*}
\end{lemma}
\begin{proof}
The probability that any given pair of locations in a 1-exRec is faulty is
bound by $\epsilon^2$, because the faults are independent. Thus the probability
that a 1-exRec is bad, i.e., contains at least two faults, is $\epsilon^{(1)}
\leq A\epsilon^2$. Similarly a $k$-exRec is bad if it contains at least two bad
$(k-1)$-exRecs. The events of any two $(k-1)$-exRecs being bad is also
independent, so the probability that a $k$-exRec is bad is $\epsilon^{(k)} \leq
A(\epsilon^{(k-1)})^2$. Solving this recursion gives us the desired bound.
\end{proof}

We can improve this bound by noting that if a $k$-exRec contains two bad
$(k-1)$-exRecs that are not independent, the $k$-exRec can still be considered
good. The analysis required to arrive at such a better bound is carried out
in \cite[Section 5.2.1]{aliferis_2005}. It is also worth noting that some
pairs of locations are benign, in the sense that if such a pair is faulty the
$k$-exRec can still be correct. A sharper bound can be obtained by not counting
such pairs, see \cite[Section 6]{aliferis_2005} for the revised argument.

It is still clear from the result presented here that if $\epsilon < 1/A$,
then the expected number of bad exRecs decreases doubly exponentially as
$k$ increases. We will now use this consequence of the lemma to prove the
threshold result.

The threshold result will show that we can reduce the \indef{computation error}
of any quantum computation to below an arbitrarily small amount. To formulate
the theorem we must first define what we mean by the computation error. Given the
probability distributions $P = \{p_i\}$ and $P' = \{p'_i\}$ of the measurements
of two quantum computations, we define the \indef{$L_1$-distance} between them
as $\sum_i{\norm{p_i - p'_i}}$. Note that if $P = P'$, then the $L_1$-distance
between them is 0. The computation error of a quantum computation is now defined
as the $L_1$-distance between that computation and the ideal computation.

\begin{theorem}[{\cite[Theorem 1]{aliferis_2005}}]
\label{ali1}
Assume conditions 1-5. Let $A$ be the largest number of pairs of locations in
any 1-exRec and assume a noise model where faults occur at a location with
probability $\epsilon$ independently. If $\epsilon < 1/A$, then for any $\delta$
there exists a level $k$ such that $M_k$ simulates a given circuit $M_0$ with
computation error at most $\delta$.
\end{theorem}
\begin{proof}
Let $P^\textrm{(ideal)}$ be the probability distribution of the outcome of a
measurement of the ideal circuit $M_0$ and let $P^\textrm{(actual)}$ be that of
a measurement of the circuit $M_k$. We define $\delta$ to be the $L_1$ distance
between these distributions, i.e., $\delta := \sum_i \norm{p_i^\textrm{(actual)}
- p_i^\textrm{(ideal)}}$. Let $L$ be the number of locations in $M_0$ and note
that each such location is encoded in $M_k$ by a $k$-Rec.

Our computation will succeed if there are no bad $k$-exRecs $M_k$, but it might
fail if there are. There are $L$ $k$-Recs in $M_k$ and by Lemma \ref{ali2} the
probability that a $k$-Rec is bad is bound by $\epsilon^{(k)}$. So by the
union bound we have that the probability that at least one $k$-Rec in $M_k$ is
bad is
\begin{equation*}
P^{(k)}_\textrm{fail} \leq L\epsilon^{(k)} \leq \frac{L(A\epsilon)^{2^k}}{A}.
\end{equation*}
Let us call the averaged probability distribution over the outcomes of
computations with at least one bad $k$-Rec $P^\textrm{(fail)}$. Naturally by
Lemma \ref{ali4} we know that if all $k$-exRecs are good, then $P^\textrm{
(actual)} = P^\textrm{(ideal)}$. This lets us write
\begin{equation*}
p_i^\textrm{(actual)} = (1 - P^{(k)}_\textrm{fail})p_i^\textrm{(ideal)} +
P^{(k)}_\textrm{fail}p_i^\textrm{(fail)}.
\end{equation*}
Which gives us
\begin{align*}
\delta &= \sum_i \norm{p_i^\textrm{(actual)} - p_i^\textrm{(ideal)}} \\
&= (1 - P^{(k)}_\textrm{fail}) \sum_i \norm{p_i^\textrm{(ideal)} -
  p_i^\textrm{(ideal)}} + P^{(k)}_\textrm{fail}\norm{p_i^\textrm{(fail)} -
  p_i^\textrm{(ideal)}} \\
&= 0 +  P^{(k)}_\textrm{fail}\norm{p_i^\textrm{(fail)} - p_i^\textrm{(ideal)}}\\
&\leq 2P^{(k)}_\textrm{fail} \\
&\leq \frac{2L(A\epsilon)^{2^k}}{A},
\end{align*}
where the first inequality is because the maximum $L_1$ distance between any two
probability distributions is 2. We can rewrite this inequality to see that we
can pick $k$ such that
\begin{equation*}
2^k \geq \frac{\log(\frac{2 L}{\delta A})}{\log(\frac{1}{\epsilon A})}
\end{equation*}
to achieve an error less than or equal to $\delta$.
\end{proof}
The crucial observation is that $k$ scales at about $\log\log(1 /\delta)$, meaning
that the error can be doubly exponentially reduced by only increasing the level
of the simulation linearly.

We have already hinted at several possible optimizations for this result by
refining the analysis, such as counting only pairs of independent bad $k$-exRecs.
In this paper we only consider QECCs that can correct 1 error, but there are
also QECCs that can correct more errors. Building fault-tolerant computers
using such QECCs can yield better thresholds, see, e.g., \cite{aliferis_2005}
or \cite{paetznick}. Another way to bring fault-tolerant quantum computing
forward is by showing that threshold results exist for a wide variety of
different noise models. This section dealt exclusively with the simplest
possible noise model in terms of analysis; stochastically independent noise. In
the following section we will show that threshold results can also be obtained
for slightly less favorable noise models.

%
%
\section{A threshold result for local non-Markovian noise}
\label{nonmarkov}
One of the properties that makes the independent noise model we have looked
at so far easy to analyze, is that the errors (or faults) it introduces are
not correlated in space or in time. In other words, an error is just as
likely to occur at a location close to where another error occurs as
it is to occur anywhere else in the circuit. The same goes for temporal
correlations: there are none. These restrictions may be physically unrealistic.
Removing them and allowing noise to be correlated in time and space is a first
step towards adversarial noise models. This section presents a threshold result
shown in \cite[Section 11]{aliferis_2005} and \cite{terhal_2005} for such a
noise model.

In the previous sections we have looked at quantum circuits as closed systems,
somehow isolated from their environment. The evolution of any closed quantum
system over time is governed by the Schr\"odinger equation,
\begin{equation*}
i \hbar \frac{d \ket{\phi(t)}}{d t} = H \ket{\phi(t)},
\end{equation*}
where $\hbar$ is Planck's constant and $H$ is a Hermitian matrix called the
\indef{Hamiltonian} of the system. The Hamiltonian is the observable for the
total energy of the system. We assume for now that $H$ is time-independent,
but it is also possible to reason about time-dependent Hamiltonians, i.e.,
Hamiltonians that are parametrized by a time variable. When we reason
about small enough time intervals though, we can treat the Hamiltonian for
a single time interval as one that is not time-dependent. To compute the
evolution of the system over some time period, say for the time interval
$(t_1, t_2)$, we solve the Schr\"odinger equation to obtain
\begin{equation*}
\ket{\phi(t_2)} = e^{\frac{-i (t_2 - t_1) H}{\hbar}} \ket{\phi(t_1)}.
\end{equation*}
It is customary to absorb $1 / \hbar$ into $H$, so we can write
\begin{equation*}
\ket{\phi(t_2)} = e^{-i (t_2 - t_1) H} \ket{\phi(t_1)}.
\end{equation*}
By linear algebra it can be shown that if $H$ is a Hermitian matrix, then
$e^{i H}$ is a unitary matrix. In particular $e^{-i (t_2 - t_1) H} =
U(t_2, t_1)$ for some unitary operator $U(t_2, t_1)$, which is called a
\indef{time-evolution operator}. This explains why we can model quantum
computing by quantum  circuits consisting of unitary operators. See, e.g.,
\cite[Chapter 2.2.2]{nielsen} for more details about the relation between
the Schr\"odinger equation and the quantum circuit model.

As opposed to the previous section we will now take the system's environment,
which we also refer to as a \indef{bath}, into account. Because the Schr\"odinger
equation only applies to closed quantum systems, this means that the
Hamiltonian must also describe the evolution of the bath and the interaction
between the bath and the system, i.e., our circuit. In general the Hamiltonian
for of the system and bath may be time-dependent. We may express the
time-dependent Hamiltonian $H(t)$ of our system and bath as
\begin{equation*}
H(t) = H_S(t) + H_{SB}(t) + H_B(t),
\end{equation*}
where $H_S(t)$ is the Hamiltonian for the system in isolation, $H_B(t)$ that
for the bath in isolation and $H_{SB}(t)$ that for the interaction between the
system and the bath. We refer to the latter as the \indef{interaction
Hamiltonian}. So far we have not placed any restrictions on the noise. We
limit the power of the noise by requiring that the interaction Hamiltonian
has the form
\begin{equation*}
H_{SB}(t) = \sum_{a \in A_t} H_{SB,a},
\end{equation*}
where each $a \in A_t$ is a set of qubits that are acted upon by the same gate
in the circuit at time $t$. For example if $q_1$ and $q_2$ are acted upon by a
CNOT gate at time $t$, then $\{q_1, q_2\} \in A_t$. Another example would be if
a qubit $q_3$ is resting at time $t'$ (remember this is equivalent to an $I$
gate acting on it), then $\{q_3\} \in A_{t'}$. We call a pair $(a, t)$ such
that $a \in A_t$ a \indef{microlocation}. Note that this does
\emph{not} correspond to what we called a location in Section \ref{frame}, but
that a location in that sense does consist of a number of microlocations as
described here.

This restriction on the interaction Hamiltonian limits the power of the noise
in the sense that errors can only be correlated when the qubits involved are
already being correlated by the circuit. Since each gate in the circuit typically
operates on few qubits (1, 2 or 3) for short periods of time, this model allows
for weak spatial and temporal correlations. Long-range correlations (both in
space and time) are still possible, because the interaction Hamiltonian can
move information from one place in time (space) to another via the bath. We
will see, however, that the influence of such indirect correlations does not
stand in the way of obtaining a threshold. Still the bath can be seen as having
a `memory', even though the noise it produces is highly localized in nature.
Informally we can say that a `memoryless' process is a Markovian process.
This is why we refer to this noise model as \indef{local non-Markovian noise},
because it is a process that has a `memory' and acts locally.

To aid our analysis we discretize the evolution of our system. Say that the
total time required for our computation is $T$ and let $t_0$ be the fundamental
gate. We now divide our total time $T$ into $N$ intervals of length $\Delta$
such that $t_0 \gg \Delta$, where $t_0 / \Delta$ is integer. This last condition
allows us to safely ignore factors of $O(\Delta^2)$ in our analysis. Because
$\Delta$ is so small, we can act as though the $H(t')$ for $t \leq t' \leq t +
\Delta$ are all approximately equal to $H(t)$. We can then, using the Trotter
expansion, express the evolution of our system for the time interval $(t,
t+\Delta)$ as
\begin{align*}
U(t+\Delta, t) &\approx e^{-i\Delta H(t)} \\
&\approx e^{-i\Delta H_S(t)} e^{-i\Delta H_B(t)} e^{-i\Delta H_{SB}(t)} \\
&= e^{-i\Delta H_S(t)} e^{-i\Delta H_B(t)} e^{-i\Delta \sum_{a \in A_t}
  H_{SB,a}} \\
&\approx e^{-i\Delta H_S(t)} e^{-i\Delta H_B(t)} \prod_{a \in A_t} e^{-i\Delta
  H_{SB,a}}.
\end{align*}
Note that we have disregarded terms with norm of $O(\Delta^2)$. It can be shown
that for small enough $\Delta$, the error in this approximation is small enough
for our purposes. Using a Taylor expansion we can further simplify this to
\begin{equation*}
U(t+\Delta, t) \approx e^{-i\Delta H_S(t)} e^{-i\Delta H_B(t)} \prod_{a \in A_t}
(I - i\Delta H_{SB,a}).
\end{equation*}
We can express the entire evolution of the system as a product of $N$ such time
evolution operators, each for a time interval of length $\Delta$. Writing out
this product, we obtain a sum where in each summand the interaction part will
contain $I$ for some microlocations and $-i\Delta H_{SB,a}$ for others.
We call the whole sum the \indef{fault-path decomposition} of the computation
and refer to a single summand as a \indef{fault path}. When a factor $-i\Delta
H_{SB,a}$ occurs in some fault path at some microlocation we say that
the fault path has a fault at that microlocation.

In Section \ref{frame} we proved a threshold result by showing that if exRecs
were good, then they were correct; that correct exRecs yield good final answers
to computations and that there were few bad exRecs. In this section we will not
only reason about exRecs, but also about fault paths. We will show that the
norm of the sum over bad fault paths, i.e., fault paths with many faults, can
be made arbitrarily small. Then we show that a small norm of bad fault paths
lead to approximately correct answers for the computation. From that we will
conclude with a threshold theorem for local non-Markovian noise.

First we still need to define the strength of such noise, cf. the error-rate 
$\epsilon$ from Section \ref{frame}. We express this in terms of the
\indef{norm} of the interaction Hamiltonian. The norm of an operator $A$ is
defined as
\begin{equation*}
\snorm{A} = \sup_{\ket{\phi}}\frac{\snorm{A\ket{\phi}}}{\snorm{\ket{\phi}}},
\end{equation*}
where $\snorm{\ket{\phi}}$ is the Euclidean norm of a state $\ket{\phi}$,
i.e., $\snorm{\ket{\phi}} = \sqrt{\braket{\phi}{\phi}}$. We shall make use of
the following properties of the norm. For all operators $A$ and
$B$ we have
\begin{equation*}
\snorm{A + B} \leq \snorm{A} + \snorm{B} \qquad \text{and} \qquad \snorm{AB}
\leq \snorm{A}\cdot\snorm{B} = \snorm{A \otimes B}.
\end{equation*}
Let $\lambda_0$ be an upper bound for the norm of the $H_{SB,a}$ Hamiltonians.
In other words, for all times $t$ and all $a \in A_t$ we have $\snorm{H_{SB,a}}
\leq \lambda_0$.

Because we picked $\Delta$ such that $t_0 / \Delta$ is an integer and $t_0
\gg \Delta$, the time spent executing a single gate is divided into many
microlocations, which we can group into so-called \indef{locations}. Note
that these locations are the same as the locations we considered in Section
\ref{frame}. Given a set $\mathcal{I}_R$ of $r$ locations we let $E(
\mathcal{I}_R)$ be the sum of all fault paths with faults at all of the $r$
locations in $\mathcal{I}_R$. If for all $r$ and all $\mathcal{I}_R$ with
$\norm{\mathcal{I}_R} = r$ we have $\snorm{E(\mathcal{I}_R)} \leq \eta^r$,
we call $\eta$ the \indef{noise strength}. The motivation for this definition
of noise strength is that in this model the noise is caused by energy being
transferred from the system to the bath or vice versa. The strength of these
interactions thus determines the strength of the noise.

As in Section \ref{frame} we can reason about locations being hit by a fault or
``being faulty''. We say that a location is hit by a fault if at least one of
the microlocations that it is made up of is faulty. As before we call our
ideal circuit $M_0$ and replace all locations in $M_0$ by 1-Recs to obtain
$M_1$ and so on. Note that each fault path describes a quantum evolution and
can therefore be seen as a circuit in itself. This allows us to say that a
fault path of $M_1$ is \indef{good} if each 1-Rec contains at most 1 fault
and it is \indef{bad} otherwise. In general a fault path of $M_k$ is good if
each $k$-Rec contains at most one bad $(k-1)$-Rec and it is bad otherwise.
We need to bound the norm of the sum of all bad fault paths of $M_k$. We
start by considering how to bound the norm of the sum of all fault paths
with at least 2 faults in $M_0$. From this we will obtain a bound on the
sum of all bad fault paths in $M_1$ and from that our desired bound on the
sum of the bad fault paths in $M_k$.

Let $F$ denote the sum of all fault paths with at least 2 faults in $M_0$. We
would like to express $F$ as a sum of $E(\mathcal{I}_R)$s, because we have a
bound for their norms. Let $\mathcal{I}$ denote the set of all locations
in $M_0$. We cannot say that $F = \sum_{r = 2}^{\norm{\mathcal{I}}}\sum_{
\mathcal{I}_R \subseteq \mathcal{I}: \norm{\mathcal{I}_R} = r} E(\mathcal{I}_R)$,
because we would be massively overcounting the number of bad fault paths. To see
this note that a fault path with faults at all $\norm{\mathcal{I}}$
locations is counted for every value of $r$. To properly count $F$ we need
a combinatorial lemma.

\begin{lemma}[{\cite[Lemma 7]{aliferis_2005}}]
\label{combi}
Let $\mathcal{I}$ be the set of all locations in $M_0$, then
\begin{equation*}
F = \sum_{r = 2}^{\norm{\mathcal{I}}} (-1)^{r} (r-1) \sum_{
\mathcal{I}_R \subseteq \mathcal{I}: \norm{\mathcal{I}_R} = r} E(\mathcal{I}_R).
\end{equation*}
\end{lemma}
\begin{proof}
We must show that every fault path with at least two faults is counted exactly
once and that fault paths with less than two fault are not counted. The latter is
easy to see, because such a fault path will not be in any of the $E(\mathcal{I}_R
)$ since we start from $r=2$. To see the former, let $f$ be any fault path with
at least two faults and say that $k$ is the number of faults on $f$. First note
that $f$ does not occur in any $E(\mathcal{I}_R)$ for $r > k$. So it suffices to
show that $f$ occurs exactly once in
\begin{equation*}
\sum_{r=2}^k (-1)^r (r-1) \sum_{\mathcal{I}_R \subseteq \mathcal{I}:
\norm{\mathcal{I}_R}=r} E(\mathcal{I}_R).
\end{equation*}
Given $2 \leq r \leq k$, there are $\binom{k}{r}$ sets $\mathcal{I}_R \subseteq
\mathcal{I}$ with $\norm{\mathcal{I}_R} = r$ that are a subset of the locations
at which $f$ has faults. Only for those $\mathcal{I}_R$, $f$ is counted in
$E(\mathcal{I}_R)$ and it occurs there once. So the number of times $f$ is
counted in this sum and with that in the total sum is
\begin{align*}
\sum_{r=2}^k (-1)^r (r-1) \binom{k}{r} &= \sum_{r=2}^k (-1)^r r \binom{k}{r} -
  \sum_{r=2}^k (-1)^r \binom{k}{r} \\
&= \sum_{r=2}^k (-1)^r k \binom{k-1}{r-1} - \Bigg( \sum_{r=0}^k (-1)^r
  \binom{k}{r}- (-1)^0 \binom{k}{0} - (-1)^1 \binom{k}{1} \Bigg) \\
&= -k \sum_{t=1}^{k-1} (-1)^t \binom{k-1}{t} - (k - 1) \\
&= -k \Bigg(\sum_{t=0}^{k-1} (-1)^t \binom{k-1}{t} - (-1)^0 \binom{k-1}{0}
  \Bigg) - (k-1) \\
&= k - k + 1 = 1,
\end{align*}
where for the third and fifth equality we use the binomial theorem, which states
in particular that $\sum_{k=0}^{n} (-1)^k \binom{n}{k} = 0$.
\end{proof}

We can now compute $\snorm{F}$. Note that the norm of a sum is bounded by a sum
of norms and that there are $\binom{\norm{\mathcal{I}}}{r}$ different
$\mathcal{I}_R \subseteq \mathcal{I}$ with $\norm{\mathcal{I}_R} = r$. Letting
$\norm{\mathcal{I}} = A$ we have
\begin{align*}
\snorm{F} &\leq \sum_{r=2}^A (r-1) \binom{A}{r} \eta^r 
= \binom{A}{2} \eta^2 \sum_{r=2}^A \frac{2}{r} \binom{A-2}{r-2} \eta^{r-2} \\
&\leq \binom{A}{2} \eta^2 \sum_{t=0}^{A-2} \binom{A-2}{t} \eta^t 
= \binom{A}{2} \eta^2 (\eta + 1)^{A-2} \\
&\leq \binom{A}{2} \eta^2 e^{(A-2)\eta},
\end{align*}
where the last inequality is because $1 + \eta \leq e^\eta$. Observe that we need
to drop the negative signs from Lemma \ref{combi}, because the fault paths can
positively interfere with one another. The only thing we can safely say for any
$E(\mathcal{I}_R)$ and $E(\mathcal{I}_R')$ is that $\snorm{E(\mathcal{I}_R) -
E(\mathcal{I}_R')} \leq \snorm{E(\mathcal{I}_R)} + \snorm{E(\mathcal{I}_R')}$.
This and the assumption that the noise obeys $\snorm{E(\mathcal{I}_R)} \leq
\eta^r$ explains the first inequality.

To bound the norm of the bad fault paths in $M_1$ we will need some additional
notation. We let $\mathcal{I}_R^{(1)}$ denote a set of 1-Recs in $M_1$ with
$\norm{\mathcal{I}_R^{(1)}} = r$ and we let $E(\mathcal{I}_R^{(1)})$ be the sum
of all fault paths where all of the 1-Recs in $\mathcal{I}_R^{(1)}$ are bad.
Recall that a 1-Rec is bad if it contains at least 2 faults. For a given
$\mathcal{I}_R^{(1)}$ we can label the bad 1-Recs it contains by $b \in [r]$
according to some arbitrary ordering. We then let $\mathcal{I}(b)$ be the set
of all locations in the 1-Rec labeled by $b$. In general we define
$\mathcal{I}_R^{(k)}$ to be a set of $k$-Recs in $M_k$ with $\norm{
\mathcal{I}_R^{(k)}} = r$ and let $E(\mathcal{I}_R^{(k)})$ be the sum of
all fault paths where all of the $k$-Recs in $\mathcal{I}_R^{(k)}$ are bad. For
a given $\mathcal{I}_R^{(k)}$ we label the bad $k$-Recs by $b \in [r]$. We let
$\mathcal{I}^{(k)}(b)$ be the set of all $(k-1)$-Recs in the $k$-Rec labeled by
$b$.

\begin{lemma}[{\cite[Lemmas 8,9]{aliferis_2005}}]
\label{ali89}
If for all $r$ and all $\mathcal{I}_R$ with $\norm{\mathcal{I}_R}=r$, $\snorm{E(
\mathcal{I}_R)} \leq \eta^r$ and $\eta \leq \frac{1}{\binom{A}{2}e^{(A-2)
\eta}}$, then
\begin{equation*}
\snorm{E(\mathcal{I}_R^{(k)})} \leq \big( \eta^{(k)} \big)^r,
\end{equation*}
where
\begin{equation*}
\eta^{(k)} = \frac{\Big(\binom{A}{2}\eta e^{(A-2)\eta}\Big)^{2^k}}{\binom{A}{
2}e^{(A-2)\eta}},
\end{equation*}
$\norm{\mathcal{I}_R^{(k)}} = r$ and $A$ is the maximum number of locations in any
1-Rec.
\end{lemma}
\begin{proof}
We prove this by induction on $k$. For $k = 1$ we have that $E(\mathcal{I}_R^{
(1)})$ is the sum over all fault paths that have $\geq 2$ faults at every 1-Rec
in $\mathcal{I}_R^{(1)}$. Let $r = \norm{\mathcal{I}_R^{(1)}}$ and note that we
label the 1-Recs in $\mathcal{I}_R^{(1)}$ by $1 \leq b \leq r$. Using Lemma
\ref{combi}, $E(\mathcal{I}_R^{(1)})$ is all the fault paths that occur for every
1-Rec in $\mathcal{I}_R^{(1)}$ in the expression
\begin{equation*}
\sum_{\ell_b = 2}^{\norm{\mathcal{I}(b)}} (-1)^{\ell_b} (\ell_b - 1) \sum_{
\mathcal{J}(b) \subseteq \mathcal{I}(b) : \norm{\mathcal{J}(b)} = \ell_b}
E(\mathcal{J}(b)).
\end{equation*}
To obtain $E(\mathcal{I}_R^{(1)})$ we can thus sum over the sets of locations
in each 1-Rec independently and consider only the fault paths with faults at
all the locations in each $\mathcal{J}$. So we have
\begin{align*}
E(\mathcal{I}_R^{(1)}) &= \sum_{\ell_1=2}^{\norm{\mathcal{I}(1)}}(-1)^{\ell_1}(
  \ell_1-1) \cdots
  \sum_{\ell_r = 2}^{\norm{\mathcal{I}(r)}} (-1)^{\ell_r} (\ell_r - 1) \\
&\qquad \qquad \qquad \sum_{\mathcal{J}(1) \subseteq \mathcal{I}(1):
  \norm{\mathcal{J}(1)}= \ell_1} \cdots
  \sum_{\mathcal{J}(r) \subseteq \mathcal{I}(r): \norm{\mathcal{J}(r)}= \ell_r}
  E\Bigg(\bigcup_{i=1}^r \mathcal{J}(i)\Bigg).
\end{align*}
By our bound on the noise strength we have
\begin{equation*}
\snorm{E\Bigg(\bigcup_{i=1}^r \mathcal{J}(i) \Bigg)} \leq \eta^{\sum_{i=1}^r
\ell_i} = \prod_{i=1}^r \eta^{\ell_i},
\end{equation*}
because $\norm{\bigcup_{i=1}^r \mathcal{J}(i)} \leq \sum_{i=1}^r \ell_i$. The norm
of $E(\mathcal{I}_R^{(1)})$ now becomes
\begin{align*}
\snorm{E(\mathcal{I}_R^{(1)})} &\leq
    \sum_{\ell_1=2}^{\norm{\mathcal{I}(1)}}(-1)^{\ell_1}(\ell_1-1)
    \cdots
    \sum_{\ell_r = 2}^{\norm{\mathcal{I}(r)}} (-1)^{\ell_r} (\ell_r - 1) \\
&\qquad \qquad \qquad \sum_{\mathcal{J}(1) \subseteq \mathcal{I}(1):
    \norm{\mathcal{J}(1)}= \ell_1}
    \cdots
    \sum_{\mathcal{J}(r) \subseteq \mathcal{I}(r): \norm{\mathcal{J}(r)}= \ell_r}
    \prod_{i=1}^r \eta^{\ell_i} \\
&= \prod_{i=1}^r \sum_{\ell_i = 2}^{\norm{\mathcal{I}(i)}} (-1)^{\ell_i} (\ell_i
  - 1)
  \sum_{\mathcal{J}(i) \subseteq \mathcal{I}(i): \norm{\mathcal{J}(i)} = \ell_i}
  \eta^{\ell_i} \\
&= \prod_{i=1}^r \sum_{\ell_i = 2}^{\norm{\mathcal{I}(i)}} (-1)^{\ell_i} (\ell_i - 1)
  \binom{\norm{\mathcal{I}(i)}}{\ell_i} \eta^{\ell_i}.
\end{align*}
So by our previous analysis of $\snorm{F}$ and letting $A$ be the maximum number
of locations in any 1-Rec we obtain that
\begin{equation*}
\snorm{E(\mathcal{I}_R^{(1)})} \leq \Bigg(\binom{A}{2} \eta^2 e^{(A-2)\eta}
\Bigg)^r,
\end{equation*}
thus proving our basis case. For $k + 1$ we note that $E(\mathcal{I}_R^{(k+1)})$
is the sum over all fault paths with $\geq 2$ bad $k$-Recs in each $(k+1)$-Rec
in $\mathcal{I}_R^{(k+1)}$. Reasoning as before we thus obtain
\begin{align*}
\snorm{E(\mathcal{I}_R^{(k+1)})} &= \prod_{i=1}^r \sum_{\ell_i = 2}^{\norm{
  \mathcal{I}^{(k+1)}(i)}} (-1)^{\ell_i} (\ell_i - 1) \binom{\norm{\mathcal{I}^{(k+1)}
  (i)}}{\ell_i} {\eta^{(k)}}^{\ell_i} \\
&\leq \Bigg(\binom{A}{2} (\eta^{(k)})^2 e^{(A-2)\eta^{(k)}} \Bigg)^r.
\end{align*}
By induction hypothesis and the assumption that $\eta \leq \frac{1}{\binom{A}{2}
e^{(A-2) \eta}}$ we have
\begin{align*}
\eta^{(k)} &= \frac{\Big(\binom{A}{2}\eta e^{(A-2)\eta}\Big)^{2^k}}{\binom{A}{2}
  e^{(A-2)\eta}}
= \frac{\binom{A}{2}\eta e^{(A-2)\eta} \Big(\binom{A}{2} \eta e^{(A-2)
  \eta}\Big)^{2^{k}-1}}{\binom{A}{2} e^{(A-2)\eta}} \\
&= \eta \Bigg(\binom{A}{2} \eta e^{(A-2)\eta}\Bigg)^{2^{k}-1} 
\leq \eta \Bigg(\frac{\binom{A}{2} e^{(A-2)\eta}}{\binom{A}{2} e^{(A-2)\eta}}
  \Bigg)^{2^{k}-1} 
= \eta.
\end{align*}
Now we can simplify our bound as
\begin{equation*}
\snorm{E(\mathcal{I}_R^{(k+1)})} \leq \Bigg(\binom{A}{2} (\eta^{(k)})^2
e^{(A-2)\eta} \Bigg)^r.
\end{equation*}
When we fill in the value of $\eta^{(k)}$ obtained from the induction hypothesis
we get
\begin{equation*}
\snorm{E(\mathcal{I}_R^{(k+1)})} \leq \Bigg(\binom{A}{2} \Bigg(
\frac{\Big(\binom{A}{2}\eta e^{(A-2)\eta}\Big)^{2^k}}{\binom{A}{2}
  e^{(A-2)\eta}}
\Bigg)^2
e^{(A-2)\eta} \Bigg)^r
= \Bigg(
\frac{\Big(\binom{A}{2}\eta e^{(A-2)\eta}\Big)^{2^{(k+1)}}}{\binom{A}{2}
  e^{(A-2)\eta}}
\Bigg)^r,
\end{equation*}
thus proving the lemma.
\end{proof}

This lemma is in much the same spirit as Lemma \ref{ali2} and like with that
lemma the bound can be improved by not counting benign pairs of faulty locations,
see \cite[Section 11.5]{aliferis_2005} for details. We are now ready to prove our
threshold result for local non-Markovian noise.

\begin{theorem}[{\cite[Theorem 6]{aliferis_2005}}]
Let $A$ be the maximum number of locations in any 1-Rec and assume a noise
model where for all $r$ and all $\mathcal{I}_R$ with $\norm{\mathcal{I}_R}=r$,
$\snorm{E(\mathcal{I}_R)} \leq \eta^r$. If $\eta < \frac{1}{\binom{A}{2}e^{(A-2)
\eta}}$, then for any $\delta$ there exists a level $k$ such that $M_k$
simulates a given circuit $M_0$ with error at most $\delta$.
\end{theorem}
\begin{proof}
Let us write the fault path expansion of $M_k$ as $G_k + B_k$, where $G_k$
contains all the fault paths without any bad $k$-exRecs and $B_k$ contains all
the fault paths with at least one bad $k$-exRec. Let $L$ be the number of
locations in $M_0$ and hence the number of $k$-exRecs in $M_k$.

We claim that
\begin{equation*}
B_k = \sum_{r=1}^{L} (-1)^{r-1} \sum_{\mathcal{I}^{(k)}_R \subseteq
\mathcal{I}^{(k)}: \norm{\mathcal{I}^{(k)}_R} = r}E(\mathcal{I}^{(k)}_R).
\end{equation*}
To prove this claim we must show that every fault path with at least one
fault is counted exactly once and that fault paths without faults are not
counted. The latter is a trivial observation. To see the former, let $f$ be
any fault path with at least one fault and say that $k$ is the number of
faults on $f$. First note that $f$ does not occur in any $E(\mathcal{I}_R)$
for $r > k$. So it suffices to show that $f$ occurs exactly once in
\begin{equation*}
\sum_{r=1}^k (-1)^r \sum_{\mathcal{I}^{(k)}_R \subseteq \mathcal{I}^{(k)}:
\norm{\mathcal{I}^{(k)}_R} = r}E(\mathcal{I}^{(k)}_R).
\end{equation*}
Given $1 \leq r \leq k$, there are $\binom{k}{r}$ sets $\mathcal{I}^{(k)}_R
\subseteq \mathcal{I}^{(k)}$ of size $r$ that are a subset of the locations
at which $f$ has faults. Only for those $\mathcal{I}^{(k)}_R$, $f$ is counted
and for each of them it is counted exactly once. So the number of times $f$
is counted in this sum and hence in $B_k$ is
\begin{equation*}
\sum_{r=1}^k (-1)^{r-1} \binom{k}{r} = \sum_{r=0}^k (-1)^{r-1} \binom{k}{r}
+1 = -\Bigg(\sum_{r=0}^k (-1)^r \binom{k}{r} \Bigg) + 1 = 1,
\end{equation*}
where the last equality is a consequence of the binomial theorem. Now we can
upper bound $\snorm{B_k}$ as we did $\snorm{F}$ using Lemma \ref{ali89} as
\begin{align*}
\snorm{B_k} &\leq \sum_{r=1}^L \binom{L}{r} \big(\eta^{(k)} \big)^r = L \eta^{(k)}
  \sum_{r=1}^L \frac{1}{r} \binom{L-1}{r-1} \big(\eta^{(k)}\big)^{r-1}\\
&\leq L \eta^{(k)} \sum_{t=0}^{L-1} \binom{L-1}{t} \big(\eta^{(k)}\big)^t =
  L \eta^{(k)} (\eta^{(k)} + 1)^{L-1} \\
&\leq L \eta^{(k)} e^{(L-1)\eta^{(k)}},
\end{align*}
where the last equality is again a consequence of the binomial theorem and the
last inequality is because $\eta^{(k)} + 1 \leq e^{\eta^{(k)}}$.

It is a known fact that the $L_1$-distance between measurements of two states is
at most twice the Euclidean distance between those states. Therefore the
computation error $\delta$, which is the $L_1$-distance between the measurement
of the ideal circuit and that of $M_k$, can be bound by twice the maximum
Euclidean distance between the final state of the ideal circuit and that of
$M_k$. This distance in turn is at most $\snorm{B_k}$, therefore
\begin{equation*}
\delta \leq 2 \snorm{B_k} \leq 2 L \eta^{(k)} e^{(L-1)\eta^{(k)}} \leq 2 L
\eta^{(k)} e^{(L-1)\eta},
\end{equation*}
where the last inequality is because $\eta^{(k)} \leq \eta$ as we argued in
the proof of Lemma \ref{ali89}. We can rewrite this inequality to see that
we can pick $k$ such that
\begin{equation*}
2^k \geq \frac{\log \Bigg(\dfrac{2 L e^{(L-1)\eta}}{\binom{A}{2} \delta e^{(A-2)
\eta}} \Bigg)}{\log \Bigg( \dfrac{1}{\binom{A}{2} \eta e^{(A-2) \eta}} \Bigg)}
\end{equation*}
to achieve an error less than or equal to $\delta$.
\end{proof}
As with Theorem \ref{ali1}, $k$ scales at about $\log\log(1 /\delta)$. We might
improve the threshold by considering QECCs that can correct more that one error.
In \cite{longrange} Aharonov, Kitaev and Preskill prove a threshold result%
\footnote{In \cite{alicklong} Alicki provides comments on that result. In
\cite{alicklocal} he comments on the result by Terhal and Burkard
(\cite{terhal_2005}) that the work by Aliferis, Gottesman and Preskill
(\cite{aliferis_2005}) presented in this section is based on. These comments
mainly deal with physical objections to the proposed noise models. As a computer
scientist, this author's understanding of such objections is unfortunately too
limited to gauge their impact.}
for a non-Markovian noise model that allows interactions between arbitrary pairs
of qubits, even if they are not correlated by the circuit. They obtain their
result by bounding the norm of the interaction Hamiltonian by the inverse of
the (physical) distance between the qubits that it correlates. The details of
this result will not be discussed in detail in this paper.

%
%
\section{Objections to Fault-Tolerant Quantum Computing}
\label{object}
The first objection that comes to mind when considering FTQC is that we
may simply not be able to construct quantum computers that are shielded
from noise well enough so that the strength of the noise that makes it
through to the computer is below the proved threshold value. So the question
arises if bounds can be given for the noise strength above which FTQC is
impossible. For certain noise models such bounds have indeed been proved.
For example it was shown by Buhrman, Cleve, Laurent, Linden, Schrijver and
Unger in \cite{Buhrman06} that quantum computers cannot withstand
``depolarizing noise'' that hits with probability $p \approx 0.45$.
Depolarizing noise is noise that acts independently on single qubits and
replaces a qubit with the completely mixed state $\big(\begin{smallmatrix}
1/2 & 0 \\ 0 & 1/2 \end{smallmatrix}\big)$ with probability $p$ and leaves
it untouched with probability $1-p$. A bound of $p \approx 0.36$ was
later given by Kempe, Regev, Unger and de Wolf in \cite{Kempe08} for a
slightly weaker noise model.

This section is devoted to qualitative objections rather than such
quantitative ones. We have seen that any error on a single qubit can be
corrected by for example the Shor code. When a superposition of states is
exposed to such an error, it impacts all terms of the superposition in the
same way. We can also think of noise that only hits certain terms in a
superposition, while leaving others undisturbed. Such noise we call
\indef{controlled noise}, because it is applied only to terms that meet
certain conditions. In this section we will deal with two types of controlled
noise; controlled bit flips ($X$ errors) and controlled phase flips ($Z$ errors).
Ben-Aroya and Ta-Shma show in \cite{ben-aroya_2009} that controlled bit flips
cannot be perfectly corrected, but can be \emph{approximately} corrected.
On the other hand, they also show that controlled phase flips cannot even
be approximately corrected. We will treat these subjects in that order.

\subsection{Controlled bit flips cannot be perfectly corrected}
\label{noflip}
The well known CNOT gate, for `controlled not', acts on two qubits and flips
the second only if the first is $\ket{1}$. We generalize this definition,
allowing controlled bit flips to be conditioned on any number of qubits in
the state, in any combination. We let $[n]$ denote the set $\{1, \ldots, n\}$
and say that for any $i \in [n]$ and $S \subseteq \{0,1\}^{n-1}$ the operator
$E_{i, S}$ applies $X$ to the $i$'th qubit conditioned on the other qubits
being in $S$. More formally if $x \in \{0,1\}^n$ we define $x_{-i} \in
\{0,1\}^{n-1}$ to be $x_1, \ldots, x_{i-1}, x_{i+1}, \ldots, x_n$. We let
$X^i$ be the operator that flips the $i$'th qubit of an $n$-qubit state,
i.e., $I^{\otimes (i-1)} \otimes X \otimes I^{\otimes (n-i)}$. Now we say
that for $i \in [n]$ and $S \subseteq \{0,1\}^{n-1}$,
\begin{equation*}
E_{i, S}\ket{x} = \begin{cases}
X^i \ket{x} &\textrm{if } x_{-i} \in S \\
\ket{x} &\textrm{otherwise}.
\end{cases}
\end{equation*}
To obtain the full definition we extend the above one linearly.

We are now ready to define the classes of errors of interest to us. We start
with $\mathcal{E}_\textbf{cbit}$ for \indef{controlled bit flips}, which we
define as
\begin{equation*}
\mathcal{E}_\textbf{cbit} := \{E_{i,S} \mid i \in [n], S \subseteq
\{0,1\}^{n-1}\}.
\end{equation*}
For the proof that controlled bit flips cannot be perfectly corrected we
restrict our attention to a subset of errors called $\mathcal{E}_\textbf{
singletons}$, defined by
\begin{equation*}
\mathcal{E}_\textbf{singletons} := \{E_{i,\{s\}} \mid i \in [n], s \in
\{0,1\}^{n-1}\}. 
\end{equation*}

To be able to rigorously formulate our theorems we must first have a proper
definition for what it means that a QECC is able to correct an error. We let
$L(\mathcal{N}, \mathcal{N}')$ denote the set of all linear operations from
the space $\mathcal{N}$ to the space $\mathcal{N}' \subseteq \mathcal{N}$.
\begin{definition}
\label{def21}
A QECC $\mathcal{M}$ corrects $\mathcal{E} \subset L(\mathcal{N},
\mathcal{N}')$ if for any two operators $A,B \in \mathcal{E}$ and any two
codewords $\phi,\psi \in \mathcal{M}$,
\begin{equation*}
\braket{\phi}{\psi} = 0 \Rightarrow \bra{\phi}A^*B\ket{\psi} = 0.
\end{equation*}
\end{definition}

During the proof of the first theorem we will make use of the following lemma
which gives us a necessary condition for a QECC being able to correct an
error. This lemma is a weaker version of Fact 2.1 in \cite{ben-aroya_2009}.
\begin{lemma}
\label{fact21}
If a code $\mathcal{M}$ corrects $\mathcal{E} \subset L(\mathcal{N},
\mathcal{N'})$, then for any two operators $A,B \in \mathcal{E}$ and any two
codewords $\phi,\psi \in \mathcal{M}$,
\begin{equation*}
\bra{\phi}A^*B\ket{\phi} = \bra{\psi}A^*B\ket{\psi} . 
\end{equation*}
\end{lemma}
\begin{proof}
Let $A,B$ be any two operators in $\mathcal{E}$ and let $\phi_1, \phi_2$ be
two basis vectors of $\mathcal{M}$. If $\phi_1 = \phi_2$ then we already have
$\bra{\phi_1}A^*B\ket{\phi_1} = \bra{\phi_2}A^*B\ket{\phi_2}$, so assume
$\phi_1 \neq \phi_2$. Since these are basis vectors we have that $\phi_1 +
\phi_2$ and $\phi_1 - \phi_2$ are codewords themselves and also orthogonal
(we omit the normalization for brevity). By Definition \ref{def21} we thus
have that
\begin{equation*}
\bra{\phi_1 + \phi_2}A^*B\ket{\phi_1 - \phi_2} = 0
\end{equation*}
and from that
\begin{equation*}
\bra{\phi_1}A^*B\ket{\phi_1} - \bra{\phi_1}A^*B\ket{\phi_2} +
  \bra{\phi_2}A^*B\ket{\phi_1} - \bra{\phi_2}A^*B\ket{\phi_2} = 0.
\end{equation*}
Again from Definition \ref{def21} we have that $\bra{\phi_1}A^*B\ket{\phi_2}
= 0 = \bra{\phi_2}A^*B\ket{\phi_1}$. Thus we obtain
\begin{equation*}
\bra{\phi_1}A^*B\ket{\phi_1} = \bra{\phi_2}A^*B\ket{\phi_2}.
\end{equation*}
Now let $\{\phi_i\}$ be some orthonormal basis of $\mathcal{M}$ and let $\phi
= \sum_i \alpha_i \ket{\phi_i}$ be some codeword in $\mathcal{M}$. We compute
that
\begin{equation*}
\bra{\phi}A^*B\ket{\phi} = (\sum_i \alpha_i^* \bra{\phi_i})A^*B(\sum_i
  \alpha_i \ket{\phi_i})
= \sum_{i,j} \alpha_i^*\alpha_j \bra{\phi_i}A^*B\ket{\phi_j}.
\end{equation*}
Note that for $i \neq j$ we have $\braket{\phi_i}{\phi_j} = 0 =
\bra{\phi_i}A^*B\ket{\phi_j}$ by Definition \ref{def21}, so we are left with
\begin{equation*}
\bra{\phi}A^*B\ket{\phi} = \sum_i \alpha_i^* \alpha_i \bra{\phi_i}A^*B
\ket{\phi_i} = \sum_i \norm{\alpha_i}^2 \bra{\phi_i}A^*B\ket{\phi_i}.
\end{equation*}
As we have shown $\bra{\phi_i}A^*B\ket{\phi_i} = \bra{\phi_j}A^*B\ket{\phi_j}$
for any two basis vectors $\phi_i, \phi_j$ of $\mathcal{M}$. Therefore letting
$\phi_1$ be some arbitrary basis vector of $\mathcal{M}$ we now arrive at
\begin{equation*}
\bra{\phi}A^*B\ket{\phi} = \bra{\phi_1}A^*B\ket{\phi_1} \cdot \sum_i
\norm{\alpha_i}^2 = \bra{\phi_1}A^*B\ket{\phi_1},
\end{equation*}
because $\phi$ is a codeword and thus assumed to be unitary.
\end{proof}

Now we are ready for the first negative result. We remind the reader that
a QECC $\mathcal{M}$ with $\dim(\mathcal{M}) = 2^0$ is a space consisting
of a single vector. By definition this means that such a code cannot be
used to encode two different qubits, say $\ket{0}$ and $\ket{1}$, because
their encoded versions would be identical, i.e., $\ket{\overline{0}} =
\ket{\overline{1}}$, and neither is recoverable. In particular such a QECC
would not be able to correct any errors and as such cannot play a role in
bringing about fault-tolerant quantum computing.
\begin{theorem}[{\cite[Theorem 3.1]{ben-aroya_2009}}]
\label{nobit}
There is no QECC $\mathcal{M}$ with $\dim(\mathcal{M}) > 2^0$ that can correct
$\mathcal{E}_\textbf{singletons}$.
\end{theorem}
\begin{proof}
Assume there is a QECC $\mathcal{M}$ with $\dim(\mathcal{M}) > 2^0$ that can
correct $\mathcal{E}_\textbf{singletons}$ and let $\phi = \sum_{i \in \{0,1\}^n}
\phi(i) \ket{i}$ and $\psi = \sum_{i \in \{0, 1\}^n} \psi(i) \ket{i}
$ be two
orthonormal codewords of $\mathcal{M}$. Fix $i \in [n]$ and $q \in \{0, 1\}^n$.
Denote $E = I^*E_{i, \{q_{-i}\}}
$ and $q' = q \oplus e_i$.

We can now compute that
\begin{align*}
\bra{\phi}E\ket{\phi} &= \bra{\phi}(\ket{\phi} + (\phi(q')-\phi(q))\ket{q} +
  (\phi(q)-\phi(q')) \ket{q'}) \\
&= \braket{\phi}{\phi} + \bra{\phi}((\phi(q')-\phi(q))\ket{q} + (\phi(q) -
  \phi(q')) \ket{q'}) \\
&= 1 + \phi(q)^*(\phi(q')-\phi(q)) + \phi(q')^*(\phi(q)-\phi(q')) \\
&= 1 - \phi(q)^*(\phi(q)-\phi(q')) + \phi(q')^*(\phi(q)-\phi(q')) \\
&= 1 - (\phi(q)^* - \phi(q')^*)(\phi(q) - \phi(q')) \\
&= 1 - (\phi(q) - \phi(q'))^*(\phi(q) - \phi(q')) \\
&= 1 - \norm{\phi(q) - \phi(q')}^2.
\end{align*}
Analogously we obtain that $\bra{\psi}E\ket{\psi} = 1 - \norm{\psi(q) -
\psi(q')}^2$. Furthermore
\begin{align*}
\bra{\phi}E\ket{\psi} &= \bra{\phi}E\ket{\psi} - \braket{\phi}{\psi} \\
&= \bra{\phi}(E\ket{\psi} - \ket{\psi}) \\
&= \bra{\phi}((\psi(q') - \psi(q))\ket{q} + (\psi(q) - \psi(q'))\ket{q'}) \\
&= \phi(q)^*(\psi(q') - \psi(q)) + \phi(q')^*(\psi(q) - \psi(q')) \\
&= -\phi(q)^*(\psi(q) - \psi(q')) + \phi(q')^*(\psi(q) - \psi(q')) \\
&= -(\phi(q)^* - \phi(q')^*)(\psi(q) - \psi(q')) \\
&= -(\phi(q) - \phi(q'))^*(\psi(q) - \psi(q')).
\end{align*}

Because $\braket{\phi}{\psi} = 0$ it follows from Definition \ref{def21} that
$\bra{\phi}E\ket{\psi} = 0$ as well. So either $(\phi(q) - \phi(q')) = 0$ or
$(\psi(q) - \psi(q'))=0$. In case of the former we have that $\phi(q) =
\phi(q')$. In case of the latter suppose for contradiction that $\phi(q) \neq
\phi(q')$, then $\psi(q) = \psi(q')$. From this we may conclude that
$\bra{\psi}E\ket{\psi} = \braket{\psi}{\psi} = 1$, while $\bra{\phi}E\ket{\phi}
\neq 1$, which contradicts Lemma \ref{fact21}, letting $A = I$ and $B = E_{i,
\{q_{-i}\}}
$. So for any codeword $\phi$ and any $i \in [n]$ we have that
$\phi(q) = \phi(q \oplus e_i)$. So all codewords are completely uniform and
hence there is only one codeword, up to multiplication by a scalar. This
contradicts the assumption that there is a QECC $\mathcal{M}$ with $\dim(
\mathcal{M}) > 1$ that can correct $\mathcal{E}_\textbf{singletons}$.
\end{proof}

We can give the following intuition for this result. Consider the state
$\ket{\phi} = \frac{1}{\sqrt{2}}(\ket{00} + \ket{11})$, i.e., the EPR pair, and
let $E = E_{1,\{0\}}
$. We will need at least 2 ancilla qubits to write down
the error syndrome for this state. An $X$ error could hit either qubit and
it is possible that no error occurs. We let $\ket{\textsf{synd}(X^1)}$ denote
the state of the ancillas when an $X$ error has been detected on the first
qubit and $\ket{\textsf{synd}(I)}$ denote that when no error has occurred. Now
$E\ket{\phi} = \frac{1}{\sqrt{2}}(\ket{10}\otimes\ket{\textsf{synd}(X^1)} +
\ket{11}\otimes\ket{\textsf{synd}(I)})$. Measuring the syndrome will now cause
the state to collapse to either $\ket{10}\otimes \ket{\textsf{synd}(X^1)}$ or
$\ket{11}\otimes\ket{\textsf{synd}(I)}$ and after error-correction we are left
with either $\ket{00}$ or $\ket{11}$.

This is very different indeed from the case where the error itself is a
linear combination of Pauli errors, for example let $E = \frac{1}{\sqrt{2}}(X^1
+ Z^2)$ and consider $E\ket{\phi}$. Measuring the syndrome makes the state
collapse to either $X^1\ket{\phi}\otimes\ket{\mathsf{synd}(X^1)}$ or
$Z^2\ket{\phi}\otimes \ket{\textsf{synd}(Z^2)}$. In either case $\phi$ is
recoverable by applying $X^1$ or $Z^2$ respectively.

%
%
\subsection{Controlled bit flips can be approximately corrected}
\label{approxflip}
Even though we have seen that controlled bit flips cannot be perfectly
corrected, not all is lost. Fortunately Ben-Aroya and Ta-Shma also proved in
\cite{ben-aroya_2009} that they can be \emph{approximately} corrected. This result
provides a positive intermezzo in a section otherwise devoted to negative
results. To present the proof we must first define what we mean by a QECC
approximately correcting an error.
\begin{definition}
Given a QECC $\mathcal{M}$ and $\mathcal{E} \subset L(\mathcal{N},
\mathcal{N}')$ we say that $\mathcal{M}$ is $(\mathcal{E}, \epsilon)$ immune if
for any operator $A \in \mathcal{E}$ and any $\phi \in \mathcal{M}$,
\begin{equation*}
\norm{\bra{\phi}A\ket{\phi}} \geq (1-\epsilon) \braket{\phi}{\phi},
\end{equation*}
where we call $\epsilon$ the approximation error.
\end{definition}
The motivation for this definition is that for small enough $\epsilon$,
$A\ket{\phi} \approx \ket{\phi}$, so there is no need for any kind of active
error correction. Our goal will be to show that there is a QECC $\mathcal{M}$
with high dimension that can approximately correct $\mathcal{E}_\textbf{cbit}$
with a low approximation error. We have seen in the proof of Theorem \ref{nobit}
that to \emph{perfectly} correct controlled bit flips the vectors in the
code space must be very uniform. Unfortunately to make them uniform enough
to perfectly correct the errors means limiting $\dim(\mathcal{M})$ to 1. Our
job will therefore be to conduct a careful balancing act between the uniformity
of the vectors in the code space and the dimension of the QECC.

To reason about the uniformity of vectors we first look at the \indef{influence
of variables on functions}. For a function $f: \{0,1\}^m \to \mathbb{C}$ we can
ask what the influence of a particular bit of the input is on the output. We
define the influence of the $i$'th bit as $I_i(f) = \pex_{x \in \{0,1\}^m}
\norm{f(x) - f(x \oplus e_i)}^2$, where $e_i$ is the $i$'th vector in the
standard basis. The \indef{influence of a function} is then defined as $I(f) =
\max_{i \in [m]}I_i(f)$. We can view $f$ as a vector $\sum_{x \in \{0,1\}^m}
f(x)\ket{x}$. Observe that a low influence of $f$ corresponds to uniformity
of the vector representation of $f$.

Now we begin the construction of our QECC $\mathcal{M}$. Let $B$ be an
integer such that $2B$ divides $n$ and define $n' = \frac{n}{2B}$. Fix some
balanced function $f: \{0,1\}^{n'} \to \{\pm \frac{1}{2}\}$ with low
influence $s(n')$. Balanced here means that $\sum_{x \in \{0,1\}^{n'}}
f(x) = 0$. Thus we have $I(f) = s(n')$. We can view a bitstring $x$ of length
$n$ as a sequence $x_1\cdots x_B$ of $B$ strings of length $2n'$. We further
subdivide each such string into two strings of length $n'$ so that we may write
$x = x_{1,0} x_{1,1}\cdots x_{B,0}x_{B,1}$. Let $z$ be some string in
$\{0,1\}^B$. We now define $f_z: \{0,1\}^n \to \mathbb{C}$ by $f_z(x) =
\prod_{k=1}^B f(x_{k,z_k})$ and define our code $\mathcal{M} := \textrm{Span}
\{f_z \mid z \in \{0,1\}^B\}$. We are now ready to formally state the theorem.

\begin{theorem}[{\cite[Theorem 4.1]{ben-aroya_2009}}]
$\mathcal{M}$ is an $\llbracket n,B \rrbracket$ QECC that is $(\mathcal{
E_\textbf{cbit}}, 2s(n'))$ immune.
\end{theorem}

To show that $\dim(\mathcal{M}) = 2^B$ it suffices to prove that $\{f_z \mid z
\in \{0,1\}^B\}$ is an orthogonal set. So let $z, z' \in \{0,1\}^B$ with $z
\neq z'$. We need to show that $\braket{f_z}{f_{z'}} = 0$. By our choice of
$f$ we know that $f(x_{i, z_i})$ is balanced over $\{\pm \frac{1}{2}\}$. The
definition of $f_z$ as a product of $f$s then tells us that $f_z$ is balanced
over $\{\pm 2^{-B}\}$. Since $z \neq z'$ there is some $k \in [B]$ such that
$z_k \neq z_k'$. Now observe that $f(x_{k, z_k})$ and $f(x_{k, z_k'})$ depend
on non-overlapping substrings of $x$ and hence their values are independent
and uniform over $\{\pm \frac{1}{2}\}$. Therefore the pair $(f_z(x), f_{z'}(
x))$ is uniform over $(\pm 2^{-B}, \pm 2^{-B})$ and so
\begin{equation*}
\braket{f_z}{f_{z'}} = \sum_{x \in \{0,1\}^n} f_z(x)^* f_{z'}(x) = 0,
\end{equation*}
as desired.

To prove the theorem we still need to show that for all $A \in
\mathcal{E}_\textbf{cbit}$ and all $\phi \in \mathcal{M}$ we have $\norm{
\bra{\phi}A\ket{\phi}} \geq (1-2s(n')) \braket{\phi}{\phi}$. Note that it
suffices to prove that $\norm{\bra{\phi}A\ket{\phi} - \braket{\phi}{\phi}}
\leq 2s(n') \norm{\braket{\phi}{\phi}}$, because
\begin{align*}
\norm{\bra{\phi}A\ket{\phi} - \braket{\phi}{\phi}} &\leq 2s(n')
  \norm{\braket{\phi}{\phi}} &\Rightarrow \\
\norm{\braket{\phi}{\phi}} - \norm{\bra{\phi}A\ket{\phi}} &\leq 2s(n')
  \norm{\braket{\phi}{\phi}}  &\Leftrightarrow \\
\norm{\bra{\phi}A\ket{\phi}} - \norm{\braket{\phi}{\phi}} &\geq -2s(n')
  \norm{\braket{\phi}{\phi}}  &\Leftrightarrow \\
\norm{\bra{\phi}A\ket{\phi}} &\geq \norm{\braket{\phi}{\phi}} - 2s(n')
  \norm{\braket{\phi}{\phi}}  &\Leftrightarrow \\
\norm{\bra{\phi}A\ket{\phi}} &\geq (1-2s(n')) \norm{\braket{\phi}{\phi}}.
\end{align*}
The first implication is a consequence of the reverse triangle inequality,
i.e., $\norm{\braket{\phi}{\phi}} - \norm{\bra{\phi}A\ket{\phi}} \leq
\norm{\braket{\phi}{\phi} - \bra{\phi}A\ket{\phi}}$, which in turn equals
$\norm{\bra{\phi}A\ket{\phi} - \braket{\phi}{\phi}}$. We shall prove the
first inequality and hence the theorem using the following two lemmas.

\begin{lemma}[{\cite[Lemma 4.3]{ben-aroya_2009}}]
For every $\phi \in \mathcal{M}$, $A \in \mathcal{E}_\textbf{cbit}$ and $i
\in [n]$,
\begin{equation*}
\norm{\bra{\phi}A\ket{\phi} - \braket{\phi}{\phi}} \leq 2^{n-1}I_i(\phi).
\end{equation*}
\end{lemma}
\begin{proof}
We begin by observing that to prove the lemma for all $A \in
\mathcal{E}_\textbf{cbit}$ we must show it for all error operators $E_{i, S}$,
where $i \in [n]$ and $S \subseteq \{0,1\}^n$. For the remainder of this
proof we will treat vectors $\phi, \psi \in \mathcal{M}$ as functions $h,g :
\{0,1\}^n \to \mathbb{C}$. Fixing some $i$ and $S$ we have
\begin{align*}
\bra{h}E_{i,S}\ket{g} &= \sum_{x: x_{-i} \not\in S} h(x)^*g(x) +
  \sum_{x: x_{-i} \in S} h(x)^*g(x \oplus e_i) \\
&= \sum_{x \in \{0,1\}^n} h(x)^*g(x) + \sum_{x: x_{-i} \in S}
  \big(h(x)^*g(x \oplus e_i) - h(x)^*g(x) \big).
\end{align*}
We now write $x \in \{0,1\}^n$ as a tuple $(x_{-i},x_i)$. This gives us
\begin{align*}
\norm{\bra{h}E_{i, S}\ket{g} - \braket{h}{g}} &= \norm{\sum_{y \in S} \big(
  h(y, 0)^*g(y,0) - h(y,0)^*g(y,1) - h(y,1)^*g(y,0) + h(y,1)^*g(y,1) \big)} \\
&= \norm{\sum_{y \in S} \big(h(y,0)^* - h(y,1)^* \big)\big(g(y,0) - g(y,1)
  \big)} \\
&\leq \sqrt{\sum_{y \in S} \norm{h(y,0)^* - h(y,1)^*}^2}\sqrt{\sum_{y \in S}
  \norm{g(y,0) - g(y,1)}^2} \\
&\leq \sqrt{\sum_{y \in \{0,1\}^{n-1}} \norm{h(y,0)^* - h(y,1)^*}^2}\sqrt{
  \sum_{y \in \{0,1\}^{n-1}}\norm{g(y,0) - g(y,1)}^2} \\
&= \sqrt{2^{n-1}I_i(h)}\sqrt{2^{n-1}I_i(g)},
\end{align*}
where the first inequality follows from the Cauchy-Schwarz inequality. Letting
$g = h$ and observing that we picked $i$ and $S$ arbitrarily, this proves the
lemma.
\end{proof}

\begin{lemma}[{\cite[Lemma 4.4]{ben-aroya_2009}}]
For every $\phi \in \mathcal{M}$,
\begin{equation*}
2^{n-1}I(\phi) \leq 2s(n')\norm{\braket{\phi}{\phi}}.
\end{equation*}
\end{lemma}
\begin{proof}
Fix some $i \in \{1, \ldots, n\}$ and suppose that $i$ is the $j$'th bit
in $x_{k, b}$. We are given a $\phi \in \mathcal{M}$ and write $\phi = \sum_{z \in
\{0,1\}^B}
\alpha_z f_z$. Note that if we can bound $I_i(\phi)$ for our
arbitrarily chosen $i$, then we can bound $I(\phi)$. We start with
\begin{align*}
I_i(\phi) &= \pex_{x \in \{0,1\}^n} \norm{\phi(x) - \phi(x \oplus
  e_i)}^2 \\
&= \pex_{x \in \{0,1\}^n} \norm{\sum_{z \in \{0,1\}^B}\alpha_z
  \big(f_z(x) - f_z(x \oplus e_i) \big)}^2
\end{align*}
and observe that the $f_z$ for which $f_z(x) = f_z(x \oplus e_i)$ do not
contribute to the sum. Therefore
\begin{align*}
I_i(\phi) &= \pex_{x \in \{0,1\}^n} \norm{\sum_{z: z_k = b} \alpha_z
  \big(f_z(x) - f_z(x \oplus e_i) \big)}^2 \\
&= \pex_{x \in \{0,1\}^n} \norm{\sum_{z: z_k = b} \alpha_z
  \prod_{1 \leq \ell < k} \big(f(x_{\ell, z_\ell})\big)
  \big(f(x_{k,b}) - f(x_{k,b} \oplus e_i)\big) 
  \prod_{k < \ell \leq B} \big(f(x_{\ell, z_\ell}\big)}^2.
\end{align*}
Now observe that the factor $f(x_{k,b}) - f(x_{k,b} \oplus e_i)$ does
not depend on $z$. We write $\dot{x}$ for $x$ without the $k$'th block
($x_{k, 0}$ and $x_{k, 1}$), so $\dot{x} := x_1, \ldots, x_{k-1},
x_{k+1}, \ldots, x_{B}$.
\begin{equation*}
I_i(\phi) = \pex_{x \in \{0,1\}^{n'}} \norm{f(x_{k,b}) - f(x_{k, b} \oplus
e_i)}^2 \cdot \pex_{\dot{x} \in \{0,1\}^{n - 2n'}} \norm{\sum_{z: z_k = b}
\alpha_z \prod_{1 \leq l \leq B \land l \neq k} f(x_{l,z_l})}^2.
\end{equation*}
To improve readability we define $\dot{f}_z = \prod_{1 \leq l \leq B
\land l \neq k} f(x_{l,z_l})$ and $\dot{\phi} = \sum_{z \in \{0,1\}^B}
\alpha_z \dot{f}_z$. Note that $\dot{f}_z$ does not depend on $x_{k,b}$
and neither does $\dot{\phi}$. Also note that by definition $s(n')$ is
an upper bound for the influence of $f$, so we write
\begin{equation*}
I_i(\phi) \leq s(n') \cdot \pex_{\dot{x} \in \{0,1\}^{n - 2n'}} \norm{
\dot{\phi}(\dot{x})}^2.
\end{equation*}
We had previously shown that the $f_z$ are orthogonal, so $\norm{\braket{
\phi}{\phi}}^2 = \sum_{z \in \{0,1\}^B} \norm{\alpha_z}^2 \braket{f_z}{f_z}
$. By an analogous argument the $\dot{f}_z$ are also orthogonal so $\norm{
\braket{\dot{\phi}}{\dot{\phi}}}^2 = \sum_{z: z_k = b} \norm{\alpha_z}^2
\braket{\dot{f}_z}{\dot{f}_z}$. By the uniformity of $f$ over $\{\pm
\frac{1}{2}\}$ it follows that $\braket{f_z}{f_z} = 2^n(\frac{1}{4})^B$
and $\braket{\dot{f}_z}{\dot{f}_z} = 2^{n - 2n'}(\frac{1}{4})^{B-1}$. Together
that tells us that $\braket{f_z}{f_z} = 2^{2n'-2}\braket{\dot{f}_z}{
\dot{f}_z}$. We now compute
\begin{align*}
\pex_{\dot{x} \in \{0,1\}^{n - 2n'}} \norm{\dot{\phi}(\dot{x})}^2 &=
  2^{-(n - 2n')}\norm{\braket{\dot{\phi}}{\dot{\phi}}}^2 \\
&= 2^{-(n - 2n')}\sum_{z:z_k = b}\norm{\alpha_z}^2 \braket{\dot{f}_z}{
  \dot{f}_z} \\
&= 4 \cdot 2^{-n}\sum_{z:z_k = b}\norm{\alpha_z}^2 \braket{f_z}{f_z} \\
&\leq 4 \cdot 2^{-n} \sum_{z \in \{0,1\}^B} \norm{\alpha_z}^2 \braket{
  f_z}{f_z} \\
&= 2 \cdot 2^{-(n-1)} \norm{\braket{\phi}{\phi}}.
\end{align*}
Now it immediately follows that for any $i \in [n]$ we have $2^{n-1}I_i(
\phi) \leq 2s(n')\norm{\braket{\phi}{\phi}}$ and hence $2^{n-1}I(\phi)
\leq 2s(n')\norm{\braket{\phi}{\phi}}$.
\end{proof}

In \cite{coinflip} Ben-Or and Linial define the ``tribes'' function, which splits
its $n'$-bit input into blocks of $\approx \log n' - c \log \log n'$ bits each.
Each bit is considered a Boolean variable and tribes first computes the
conjunction of the variables in each block and then outputs the disjunction of
those conjunctions. The constant $c$ is chosen so that the function becomes
approximately balanced. The influence of tribes is $O\left(\frac{\log n'}{n'}
\right)$, so when we use the tribes function to define $f(x) = 1/2$ if tribes$(x)
= 1$ and $f(x) = -1/2$ if tribes$(x) = 0$, then the theorem implies that for
every $n$ and $B$ such that $2B$ divides $n$ there is an $\llbracket n, B
\rrbracket$ QECC that is $(\mathcal{E}_\textbf{cbit}, O(\frac{B \log(n/B)}{n}))$
immune. In particular when we let $B = \sqrt{n}$, this shows that there
is an $\llbracket n, \sqrt{n} \rrbracket$ QECC that is
$(\mathcal{E}_\textbf{cbit}, O(\frac{\log(\sqrt{n})}{\sqrt{n}}))$ immune.

%
%
\subsection{Controlled phase errors cannot be approximately corrected}
\label{nophase}
We return to the negative results by considering \indef{controlled phase flips}.
For $S \subseteq \{0,1\}^n$ and $\theta \in [0, 2\pi)$ we can define the error
operator $E_{S, \theta}$ by $E_{S, \theta}\ket{x} = e^{\theta i}\ket{x}$ if
$x \in S$ and $E_{S, \theta}\ket{x} = \ket{x}$ otherwise. This lets us define
the set of error operators that are controlled phase errors as
\begin{equation*}
\mathcal{E}_\textbf{cphase} := \{E_{S, \theta} \mid S \subseteq \{0, 1\}^n
\text{ and } \theta \in [0, 2\pi)\}.
\end{equation*}

Note that for every partition $\overline{S} = (S_1,S_2,S_3,S_4)$ of $\{0,1\}^n$
the set $\mathcal{E}_\textbf{cphase}$ contains the operators $E_{S_2 \cup S_4,
-\frac{\pi}{2}}$ and $E_{S_3 \cup S_4, \pi}$. In the following we
will let $E_{\overline{S}}$ stand for $E_{S_2 \cup S_4, -\frac{\pi}{2}}^*E_{S_3
\cup S_4, \pi}$. We now observe that for all $x \in \{0,1\}^n$ we have
\begin{equation*}
E_{\overline{S}}\ket{x} =
\begin{cases}
\ket{x} &\text{if } x \in S_1 \\
e^{\frac{\pi}{2}i}\ket{x} &\text{if } x \in S_2 \\
e^{\pi i}\ket{x} &\text{if } x \in S_3 \\
e^{\frac{3\pi}{2}i}\ket{x} &\text{if } x \in S_4.
\end{cases}
\end{equation*}

We will show that there is no non-trivial QECC that can correct $\mathcal{E
}_\textbf{cphase}$. Non-trivial here means that the QECC must have more than one
codeword. In fact we will prove something stronger, namely that no non-trivial
QECC can \indef{separate} $\mathcal{E}_\textbf{cphase}$ with reasonable error.
We define this as follows.
\begin{definition}
A QECC $\mathcal{M}$ separates $\mathcal{E} \subset L(\mathcal{N},
\mathcal{N}')$ with at most $\alpha$ error if for any two operators $A,B \in
\mathcal{E}$ and any two codewords $\phi, \psi \in \mathcal{M}$,
\begin{equation*}
\braket{\phi}{\psi} = 0 \Rightarrow \norm{\bra{\phi}A^*B\ket{\psi}} \leq
\alpha.
\end{equation*}
\end{definition}
The motivation for this definition is that if a QECC cannot separate a set of
errors, then errors from that set can turn orthogonal states into non-orthogonal
states. That in turn means that they can no longer be perfectly distinguished by
any quantum measurement, and so in particular the error cannot be perfectly
corrected. To prove the theorem we first need a small lemma.

\begin{lemma}[{\cite[Lemma 5.2]{ben-aroya_2009}}]
\label{lemba52}
Let $\mathcal{M}$ be a vector space with $\dim(\mathcal{M}) > 1$. Then there
are two orthonormal vectors $\phi, \psi \in \mathcal{M}$ such that $\sum_x
\norm{\phi(x)}\cdot\norm{\psi(x)} \geq 1/2$.
\end{lemma}
\begin{proof}
Let $\phi, \psi \in \mathcal{M}$ be two orthonormal vectors and let $\phi' =
\frac{1}{\sqrt{2}}(\phi + \psi)$ and $\psi' = \frac{1}{\sqrt{2}}(\phi - \psi)$.
Then
\begin{equation*}
\sum_x \norm{\phi'(x)}\cdot\norm{\psi'(x)} = \frac{1}{2} \sum_x \norm{\phi(x) +
\psi(x)}\cdot\norm{\phi(x) - \psi(x)}.
\end{equation*}
Fixing some $x \in \{0,1\}^n$ and assuming without loss of generality that
$\norm{\phi(x)} \geq \norm{\psi(x)}$ we find
\begin{align*}
\norm{\phi(x) + \psi(x)}\cdot\norm{\phi(x) - \psi(x)} &\geq
  \norm{(\phi(x) + \psi(x))(\phi(x) - \psi(x))} \\
&= \norm{\phi(x)^2 - \psi(x)^2} \\
&\geq \norm{\phi(x)}^2 - \norm{\psi(x)}^2 \\
&\geq \norm{\phi(x)}^2 + \norm{\psi(x)}^2 - 2\norm{\phi(x)}\cdot\norm{\psi(x)},
\end{align*}
where the last inequality is because $\norm{\phi(x)} \geq \norm{\psi(x)}$. Now
we can write
\begin{equation*}
\sum_x \norm{\phi(x)'}\cdot\norm{\psi(x)'} \geq \frac{1}{2} \sum_x
(\norm{\phi(x)}^2 + \norm{\psi(x)}^2) - \sum_x \norm{\phi(x)}\cdot\norm{\psi(x)}
= 1 - \sum_x \norm{\phi(x)}\cdot\norm{\psi(x)},
\end{equation*}
so either $\sum_x \norm{\phi(x)}\cdot\norm{\psi(x)}$ or $\sum_x \norm{\phi(x)'}
\cdot\norm{\psi(x)'}$ is at least 1/2. Therefore either $\phi,\psi$ or $\phi',
\psi'$ are our witnesses.
\end{proof}
Now we are ready to prove the theorem.

\begin{theorem}[{\cite[Theorem 5.1]{ben-aroya_2009}}]
There is no QECC with dimension 2 that can separate 
$\mathcal{E}_\textbf{cphase}$ with error $\alpha \leq \frac{1}{10}$.
\end{theorem}
\begin{proof}
We need to show that for some $A,B \in \mathcal{E}_\textbf{cphase}$ and some
unitary $\phi, \psi \in \mathcal{M}$ we have both $\braket{\phi}{\psi} = 0$ and
$\norm{\bra{\phi}A^*B\ket{\psi}} > \frac{1}{10}$. So it suffices to show that
for some $\overline{S}$ we have $\norm{\bra{\phi}E_{\overline{S}} \ket{\psi}}
> \frac{1}{10}$ for some $\phi,\psi$.

Let $\phi, \psi$ be as in Lemma \ref{lemba52}. We may express
$\phi(x) = r_x e^{\theta_x i}$ and $\psi(x) = r_x' e^{\theta_x' i}$, where
$r_x = \norm{\phi(x)}$ and $r_x' = \norm{\psi(x)}$. Letting $\theta_1 = 0,
\theta_2 = \frac{\pi}{4}, \theta_3 = \frac{\pi}{2}$ and $\theta_4 = \frac{3\pi
}{4}$ we define $\overline{S} = (S_1, S_2, S_3, S_4)$ by putting $x \in S_{j(x)
}$ where $j(x) = \argmin_{j \in [4]} \{\norm{-\theta_x + \theta_x' + \theta_j}
\bmod{2 \pi}\}$. Observe that
\begin{equation*}
\min_{j \in [4]} \{\norm{-\theta_x + \theta_x' + \theta_j} \bmod{2\pi}\} \leq
\frac{\pi}{4},
\end{equation*}
because $\theta_x' - \theta_x$ modulo $2\pi$ is in $[-\pi, \pi)$.

We now let $\zeta_x = -\theta_x + \theta_x' + \theta_j$ (note that $\theta_j$
depends on $x$) and $u_x = 1 - e^{\zeta_x i}$. This lets us write
\begin{equation*}
\norm{\bra{\phi}E_{\overline{S}}\ket{\psi}} = \norm{\sum_{x \in \{0,1\}^n}
r_x r_x' e^{\zeta_x i}} = \norm{\sum_{x \in \{0,1\}^n} r_x r_x'(1 - u_x)}.
\end{equation*}
We now compute
\begin{align*}
\norm{u_x}^2 &= (1 - \cos(\zeta_x))^2 + \sin^2(\zeta_x) \\
&= 1 - 2 \cos(\zeta_x) + \cos^2(\zeta_x) + \sin^2(\zeta_x) \\
&= 2(1 - \cos(\zeta_x)).
\end{align*}
Note that $\zeta_x \in [0, \pi /4]$. Since the cosine is increasing in that
interval, $\cos(\zeta_x)$ is maximal at $\pi / 4$, where its value is $1 /
\sqrt{2}$. Hence $2(1 - \cos(\zeta_x)) \leq 2 - \sqrt{2}$ and $\norm{u_x} \leq
\sqrt{2 - \sqrt{2}}$. This lets us further derive
\begin{align*}
\norm{\sum_{x \in \{0,1\}^n} r_x r_x'(1 - u_x)} &\geq \sum_{x \in \{0,1\}^n}
  r_x r_x' - \norm{\sum_{x \in \{0,1\}^n}r_x r_x' u_x} \\
&\geq \left(1 - \max_{x \in \{0,1\}^n} \norm{u_x}\right)\sum_{x \in \{0,1\}^n}r_x r_x' \\
&\geq \left(1 - \sqrt{2 - \sqrt{2}}\right) \sum_{x \in \{0,1\}^n}r_x r_x'.
\end{align*}
By our choice of $\phi$ and $\psi$ we have that $\sum_{x \in \{0,1\}^n}
r_x r_x' \geq 1/2$, therefore
\begin{equation*}
\norm{\bra{\phi}E_{\overline{S}}\ket{\psi}} \geq \left(1 - \sqrt{2 -
\sqrt{2}}\right)\frac{1}{2} > \frac{1}{10}
\end{equation*}
as desired.
\end{proof}

Observe that the proof shows that we cannot even separate the restriction of
$\mathcal{E}_\textbf{cphase}$ to rotations over $\pi$ and $-\pi/2$, let alone the
arbitrary rotations the unrestricted $\mathcal{E}_\textbf{cphase}$ allows.

%
%
\section{More speculative objections}
\label{longtail}
In this section we present more objections to FTQC, namely a selection of
objections put forth by Kalai in \cite{kalai_detrimental_2008, kalai_propagation,
kalai_fail}. Although these objections are less precise than those presented
in the previous section, Kalai tries hard to identify possible flaws in the theory
of fault-tolerant quantum computing as it exists today. Regardless of how
well-founded these objections will turn out to be, the work of Kalai is very
valuable to help us better understand the nature of noise that impacts quantum
systems and how such noise can be guarded against.

\subsection{Noise propagation}
Most of Kalai's objections deal with the assumptions made about the noise models
for which threshold theorems have been proved. His first such objection deals with
\indef{noise propagation}, the way in which errors occurring at a particular time
during the computation spread across the circuit as the computation proceeds. In
Sections \ref{frame} and \ref{nonmarkov} we assumed that 1-Recs could be
constructed that met conditions 1 through 5 on page \pageref{cond15}, which limit
noise propagation. In particular we assumed that we could limit the accumulation
of noise, by removing all the noise every time we ran our error-detection and
error-correction procedures. Kalai suggests that modeling noise
\emph{propagation} is fundamental to modeling noisy quantum systems and that we
should identify the mathematical properties of noise propagation. In
\cite[Section 6.2]{kalai_propagation} and \cite[Section 6]{kalai_fail} Kalai
proposes such a property and conjectures that fault-tolerant quantum computing
is impossible for noise models having that property.

\subsection{Preparing codewords}
Another objection has to do with our ability to encode quantum states using a
QECC, i.e., to prepare codewords. In our discussion of fault-tolerant quantum
computing we have assumed that a qubit preparation rectangle has at most one
error in its output. In particular, this implies that the state generated by
the rectangle does not contain a superposition of codewords. Rather, it may
be a superposition of a codeword with one or more states that are not
codewords. That allowed us to perform the error-detection and error-correction
steps as outlined in Section \ref{intro}. Naturally, when the output of such a
rectangle is a superposition of codewords, we can neither detect the error nor
correct it. Kalai proposes as Conjecture 1 in \cite{kalai_fail} that the act of
preparing encoded qubits inherently results in a superposition of the intended
codeword and undesirable codewords. 

\subsection{Error synchronization}
Kalai's main objection, however, is based on a physical conjecture. For the
threshold results presented in this paper it was assumed that the spatial and
temporal correlations between errors are either non-existent (Section
\ref{frame}) or highly localized (Section \ref{nonmarkov}). When we consider
the probability distribution of the number $k$ of errors hitting an $n$-qubit
state, a consequence of that assumption becomes that beyond the expected value
for the number of errors, the probability decreases exponentially with $k$. In
other words, the probability distribution of the number of errors has a small
tail. Kalai observes that the QECCs used to prove these threshold results
generate and operate on highly entangled states. Formal definitions can be
given for measures of the entanglement of (mixed) quantum states, but we forgo
giving them here. He makes the physical conjecture that errors hitting such
entangled states will be highly correlated, a phenomenon he calls \indef{error
synchronization}. In \cite[Section 7.2]{kalai_propagation} Kalai makes this
conjecture more precise. The impact of error synchronization on FTQC can
perhaps best be understood by taking a small detour back to the classical
world and considering the effect of error correlation on binary strings. We
will prove a lemma demonstrating that this effect is that we can no longer
assume that our distributions have small tails. This is a generalization of
Lemma 1 in \cite{kalai_detrimental_2008}, which is also referenced as
Proposition 6 in \cite{kalai_propagation}. As such, if Kalai's conjecture
about error synchronization turns out to be true, this could have serious
repercussions for fault-tolerant quantum computing.

In the following we will let $[n]$ stand for the set $\{1, \ldots, n\}$. A
binary string of length $n$ can be seen as an indicator string for errors,
where a $1$ indicates that an error has occurred on that position and a $0$
indicates that no error has occurred. Given a probability distribution
$\mathcal{D}$ on binary strings $x = x_1 \cdots x_n$ of length $n$ and $i,j
\in [n]$, we define the \indef{pairwise correlation} $c_{i j}(\mathcal{D})$
to be $\prob_{x \sim \mathcal{D}}(x_j = 1 \mid x_i = 1)$. This definition
is the author's interpretation of $c_{i j}$ as it is used, but not defined
in \cite{kalai_detrimental_2008}. Note that $c_{ii}(\mathcal{D}) = 1$ for
all $\mathcal{D}$ and $i \in [n]$, but this does not matter as our lemma
will only assume a lower bound on $c_{i j}(\mathcal{D})$. For a binary string
$x$ we let $\weight{x}$ denote the Hamming weight of $x$.

\begin{lemma}
Suppose that $\mathcal{D}$ is a probability distribution on binary strings of
length $n$ and let $s$ be such that for all $i,j \in [n]$, $c_{i j}(\mathcal{D})
\geq s$. For binary strings $y$ and $z$, $y \init z$ means that $y$ is an initial
segment of $z$. Then
\[
\prob_{x \sim \mathcal{D}}(\weight{x} > s n/2) \geq \sum_{i=0}^{n-1}
  \prob(0^i1 \init x) \frac{s/2 - s i/n}{1 - s/2}.
\]
\end{lemma}

\begin{proof}
All probabilities are assumed to be according to $\mathcal{D}$, i.e., $x \sim
\mathcal{D}$. We start by observing that for $i = 0, \ldots, n-1$
\begin{equation*}
\pex(\weight{x} \mid 0^i1 \init x) \geq 1 + (n-i-1)s \geq (n-i)s,
\end{equation*}
because the $(i+1)$'st position is 1 and each of the $(n-i-1)$ positions
after that are 1 with probability $\geq s$, because the $(i+1)$'st position
is 1. Furthermore we can upper bound this expected value as
\begin{align*}
\pex(\weight{x} \mid 0^i1 \init x) \leq &\prob(\weight{x} \leq s n/2 \mid 0^i1
\init x)s n/2 + \\
&\prob(\weight{x} > s n/2 \mid 0^i1 \init x)n,
\end{align*}
for $i = 0, \ldots, n-1$. Combining these equations and using the fact that
$\prob(\weight{x} \leq s n/2 \mid 0^i1 \init x) = 1 - \prob(\weight{x} > s n/2
\mid 0^i1 \init x)$
we compute that
\begin{equation*}
\prob(\weight{x} > s n/2 \mid 0^i1 \init x) \geq \frac{s/2 - s i/n}{1 - s/2}.
\end{equation*}
Now we can use this equation together with the observation that
\begin{equation*}
\prob(\weight{x} > s n/2) = \sum_{i=0}^{n-1}\prob(0^i1 \init x) \prob(\weight{x}
> s n/2 \mid 0^i1 \init x)
\end{equation*}
to obtain the result.
\end{proof}

This lemma shows that when errors are correlated, the distribution of the
number of errors hitting a state has a fat tail, i.e., beyond the expected
number of errors the probability decreases only \emph{polynomially} with
the number of errors, not exponentially as before. Writing out the first two
terms of the bound we find that
\begin{equation*}
\prob(\weight{x} > s n/2) \geq \prob(1 \init x) \frac{s/2}{1 - s/2} +
\prob(01 \init x) \frac{s/2 - s/n}{1 - s/2},
\end{equation*}
illustrating that the decrease is indeed polynomial. This violates the
assumption we made in Section \ref{frame} that qubits are hit by errors
\emph{independently}. In Section \ref{nonmarkov} we made no such assumption,
but this result suggests that the noise strength $\eta$ might be too
large, i.e., above the threshold. Therefore if Kalai's conjecture that
highly entangled states lead to error synchronization is correct, then the
current threshold theorems do not apply for realistic noise models and the
possibility of fault-tolerant quantum computing again becomes an open question.

We observe that the bound shown in the lemma is close to optimal, for let
$\mathcal{D}$ be a distribution where $x = 0^n$ with some probability $p \in
(0,1)$ and $x = 1^n$ with probability $1-p$. Then for all $i, j \in [n]$ we
have $c_{i j}(\mathcal{D}) = 1$ and
\begin{equation*}
\prob_{x \sim \mathcal{D}}(\weight{x} > n/2) = \prob(1 \init x)
\frac{1/2}{1 -1/2} = \prob(1 \init x) = p.
\end{equation*}

%
%
\section{Conclusions and outlook}
\label{conclusion} 
The goal of this survey was to give an overview of the current state of FTQC, 
to list important positive and negative results and to show that a large, gray
area remains largely unexplored in between. In Sections \ref{frame} and
\ref{nonmarkov} we presented some important positive results, namely that
thresholds for fault-tolerant quantum computing can be established for a
number of noise models. Although the exact numerical value of these thresholds
is of great practical importance, the differences between proved minimal and
maximal values are still several orders of magnitude. Closing in on exact
numerical values for thresholds under various noise models and using various
QECCs remains an important research goal in the field of fault-tolerant
quantum computing.

The noise models for which we presented threshold results allow for only
very weak spatial and temporal correlations between errors. One direction
forward would thus be to prove that threshold results can be obtained for
more strongly correlated noise models. In fact, any relaxation of the
assumptions we made in Sections \ref{frame} and \ref{nonmarkov} would be
of great value. 

In Section \ref{object} we have seen that when we allow noise models to use
the very entanglement that gives quantum computation its edge over classical
computation against us, we should let go of the idea of perfectly correcting
errors and instead focus on \emph{approximately} correcting them. For controlled
phase errors, however, even that will not be possible. Without allowing
the noise to act conditionally in the sense of Section \ref{object}, we
can still endeavor to obtain threshold results for error distributions that
are not entirely independent. It may be possible to prove that fault-tolerant
quantum computing is possible for correlated errors, when we also consider
a threshold for the \emph{correlation} of errors.

As for obtaining more negative results, the objections put forth by Kalai
deserve further study and formalization. In the end, experiments with
actual noisy quantum systems will likely determine what will be considered
physically \emph{realistic} noise models.

\subsection*{Acknowledgments}
I would like to thank Ronald de Wolf for his guidance and great patience
in helping me understand the topics discussed in this paper. His many
comments and corrections have greatly benefited its style and presentation.
I would also like to thank Dorit Aharonov, Daniel Gottesman and Amnon
Ta-Shma for their insightful replies to Ronald's emails.

%
%
\bibliography{paper}
\end{document}